\documentclass{article} 
\usepackage{amssymb} 
\usepackage{amsmath} 
\usepackage{amsfonts} 
\usepackage{graphicx} 
\usepackage{bm} 
\usepackage{cite} 
\usepackage{setspace}
 
\newtheorem{theorem}{Theorem}[section]

\newtheorem{lemma}[theorem]{Lemma}

\newtheorem{remark}[theorem]{Remark}

\newenvironment{proof}[1][Proof]{\noindent\textbf{#1.} }{\ \rule{0.5em}{0.5em}}

\input{tcilatex}  
 
\begin{document} 
 
\title{On the existence of quasipattern solutions of the Swift--Hohenberg 
equation} 
\author{G. Iooss$^1$ \and A. M. Rucklidge$^2$ \and $^{1}${\small I.U.F., 
Universit\'e de Nice, Labo J.A.Dieudonn\'{e}} \and {\small Parc Valrose, 
F-06108 Nice, France} \and $^{2}${\small Department of Applied Mathematics, 
University of Leeds,} \and {\small Leeds LS2 9JT, England} \\ 
{\small gerard.iooss@unice.fr, A.M.Rucklidge@leeds.ac.uk}} 
\maketitle 
 
\begin{abstract} 
Quasipatterns (two-dimensional patterns that are quasiperiodic in any 
spatial direction) remain one of the outstanding problems of pattern 
formation. As with problems involving quasiperiodicity, there is a small 
divisor problem. In this paper, we consider $8$-fold, $10$-fold, $12$-fold, 
and higher order quasipattern solutions of the Swift--Hohenberg equation. We 
prove that a formal solution, given by a divergent series, may be used to 
build a smooth quasiperiodic function which is an approximate solution of 
the pattern-forming PDE up to an exponentially small error.\newline 
~\newline 
~\newline 
AMS: 35B32, 35C20, 40G10, 52C23 
\end{abstract} 
 
 
 
 
 
 
\section{Introduction} 
 
\label{sec:intro} 
 
Quasipatterns remain one of the outstanding problems of pattern formation. 
These are two-dimensional patterns that have no translation symmetries and 
are quasiperiodic in any spatial direction (see figure~\ref{fig:quasipattern}).
In spite of the lack of translation symmetry (in contrast to periodic 
patterns), the spatial Fourier transforms of quasipatterns have discrete 
rotational order (most often, $8$, $10$ or $12$-fold). Quasipatterns were first 
discovered in nonlinear pattern-forming systems in the Faraday wave 
experiment~\cite{Christiansen1992,Edwards1994}, in which a layer of fluid is 
subjected to vertical oscillations. Since their discovery, they have also 
been found in nonlinear optical systems~\cite{Herrero1999}, shaken 
convection~\cite{Volmar1997,Rogers2005} and in liquid crystals~\cite%
{Lifshitz2007a}, as well as being investigated in detail in large aspect 
ratio Faraday wave experiments~\cite%
{Binks1997,Binks1997b,Kudrolli1998,Arbell2002}. 
 
\begingroup  
 
\begin{figure}[tp] 
\begin{center} 
\includegraphics[width=0.9\hsize]{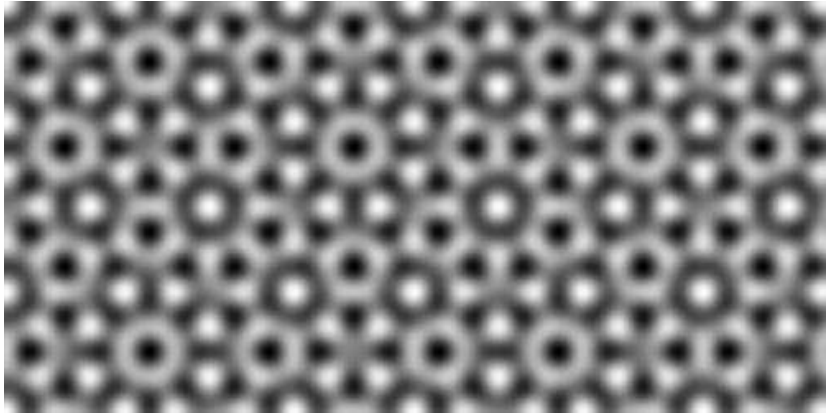} 
\end{center} 
\caption{Example $8$-fold quasipattern. This is an approximate solution of 
the Swift--Hohenberg equation~(\protect\ref{eq:shtime}) with $\protect\mu %
=0.1$, computed by using Newton iteration to find an equilibrium solution of 
the PDE truncated to wavenumbers satisfying $|\mathbf{k}|\leq \protect\sqrt{5%
}$ and to the quasilattice $\Gamma _{27}$.} 
\label{fig:quasipattern} 
\end{figure} 
 
\endgroup  
 
In many of these experiments, the domain is large compared to the size of 
the pattern, and the boundaries appear to have little effect. Furthermore, 
the pattern is usually formed in two directions ($x$ and $y$), while the 
third direction ($z$) plays little role. Mathematical models of the 
experiments are therefore often posed with two unbounded directions, and the 
basic symmetry of the problem is $E(2)$, the Euclidean group of rotations, 
translations and reflections of the $(x,y)$ plane. 
 
The mathematical basis for understanding the formation of periodic patterns 
is well founded in \emph{equivariant bifurcation theory}~\cite%
{Golubitsky1988}. With spatially periodic patterns, the pattern-forming 
problem (usually a PDE) is posed in a periodic spatial domain instead of the 
infinite plane. Spatially periodic patterns have Fourier expansions with 
wavevectors that live on a lattice. There is a parameter~$\mu $ in the PDE, 
and at the point of onset of the pattern-forming instability ($\mu =0$), the 
primary modes have zero growth rate and all other modes on the lattice have 
negative growth rates that are bounded away from zero. In this case, the 
infinite-dimensional PDE can be reduced rigorously to a finite-dimensional 
set of equations for the amplitudes of the primary modes~\cite%
{Carr1981,Chossat1994,Iooss1998a, Guckenheimer1983b,Vanderbauwhede1992}, and 
existence of periodic patterns as solutions of the pattern-forming PDE can 
be proved. The coefficients of leading order terms in these amplitude 
equations can be calculated and the values of these coefficients determine 
how the amplitude of the pattern depends on the parameter~$\mu $, and which 
of the regular patterns that fit into the periodic domain are stable. Due to 
symmetries, the solutions of the PDE are in general expressed as power 
series in~$\sqrt{\mu }$, which can be computed, and which have a non-zero 
radius of convergence. 
 
In contrast, quasipatterns do not fit into any spatially periodic domain and 
have Fourier expansions with wavevectors that live on a \emph{quasilattice} 
(defined below). At the onset of pattern formation, the primary modes have 
zero growth rate but there are other modes on the quasilattice that have 
growth rates arbitrarily close to zero, and techniques that are used for 
periodic patterns cannot be applied. These small growth rates appear as  
\emph{small divisors}, as seen below, and correspond at criticality ($\mu 
=0) $ to the fact that for the linearized operator at the origin (denoted $-%
\mathcal{L}_{0}$ below), the 0 eigenvalue is not isolated in the spectrum. 
 
If weakly nonlinear theory is applied in this case without regard to its 
validity, this results in a divergent power series~\cite{Rucklidge2003}, and 
this approach does not lead to a convincing argument for the existence of 
quasipattern solutions of the pattern-forming problem. 
 
This paper is primarily concerned with proving the \emph{existence} of 
quasipatterns as \emph{approximate} 
steady solutions of the simplest pattern-forming PDE, the 
Swift--Hohenberg equation:  
 \begin{equation} 
 \frac{\partial U}{\partial t}=\mu U-(1+\Delta )^{2}U-U^{3}  \label{eq:shtime} 
 \end{equation}
where $U(x,y,t)$ is real and $\mu $ is a parameter. We do not prove the
existence of quasipatterns as exact steady solutions of the~\hbox{PDE}. We are
not concerned with the \emph{stability} of these quasipatterns: in fact, they
are almost certainly unstable in the Swift--Hohenberg equation. Stability of a
pattern depends on the coefficients in the amplitude equations (as computed
using weakly nonlinear theory). In the Faraday wave experiment, and in more
general parametrically forced pattern forming problems, resonant mode
interactions have been identified as the primary mechanism for the stabilisation
of quasipatterns and other complex patterns (see~\cite{Rucklidge2009} and
references therein). These mode interactions are not present in the
Swift--Hohenberg equation, though their presence would not significantly alter
our results.
 
In many situations involving a combination of nonlinearity and 
quasiperiodicity, small divisors can be handled using \emph{hard implicit 
function theorems}~\cite{delaLlave2001}, of which the KAM theorem is an 
example. Unfortunately, there is as yet no successful existence proof for 
quasipatterns using this approach, although these ideas have been applied 
successfully to a range of small-divisor problems arising in other types of 
PDEs~\cite{Craig1993,Iooss2005d,Iooss2009}. There are also alternative 
approaches to describing quasicrystals based on Penrose tilings and on 
projections of high-dimensional regular lattices onto low-dimensional 
spaces~\cite{Janot1994}. 
 
We take a different approach in this paper: we show how the divergent power 
series that is generated by the naive application of weakly nonlinear theory 
can be used to generate a smooth quasiperiodic function that (a)~shares the 
same asymptotic expansion as the naive divergent series, and (b)~satisfies 
the PDE~(\ref{eq:shtime}) with an exponentially small error as $\mu$ tends 
to~$0$. This approach is based on summation techniques for divergent power 
series: see~\cite{Ramis1996,Candelpergher2006,Balser1994} for other 
examples. In order to make the paper self-contained, we put in Appendices 
some proofs of useful results, even though they are ``known''. 
 
In section~\ref{sec:Diophantine}, we define the quasilattice and derive 
Diophantine bounds for the small divisors that will arise in the nonlinear 
problem, for $Q$-fold quasilattices: Lemma~\ref{dioph_estimate} extends the 
results of~\cite{Rucklidge2003} covering the cases $Q=8$, $10$, $12$, to any 
even~$Q\geq 8$. We then compute in section~\ref{sec:formal} (following~\cite%
{Rucklidge2003}) the power series in~$\sqrt{\mu }$ for a formal 
$Q$-quasipattern solution~$U$ of the Swift--Hohenberg equation, where $\mu $ is 
the bifurcation parameter in the~\hbox{PDE}. 
 
In section~\ref{sec:spaces}, we define an appropriate function space~${%
\mathcal{H}}_{s}$: each term in the formal power series~$U$ is in this 
space. In section~\ref{sec:Gevrey}, we prove (Theorem~\ref{gevreythm}) 
bounds on the norm of each term in the formal power series solution of 
the~\hbox{PDE}. In the $Q$-fold case, the norm of the $\mu ^{n+\frac{1}{2}}$ 
term in the power series for the quasipattern is bounded by a constant 
times~$K^{n}(n!)^{4l}$, where $K$ is a constant and $l+1$ is the order of the 
algebraic number $\omega =2\cos (2\pi /Q)$, which is also half of \emph{%
Euler's Totient function}~$\varphi (Q)$ ($l=1$ for $Q=8$, $10$ and $12$, $%
l=2 $ for $Q=14$ and $18$, $l=3$ for $Q=16$, $20$, $24$, $30$, \dots ). This 
result was announced in~\cite{Rucklidge2003} for $Q\leq 12$, and is extended 
here to $Q\geq 14$. With a bound that grows in this way with~$n$, the power 
series is \emph{Gevrey-}$4l$, taking values in a space of $Q$-fold 
quasiperiodic functions.
 
In section~\ref{sec:Borel}, for convenience, we consider the cases $Q=8$, $10 
$ and~$12$. We introduce a small parameter~$\zeta $ related to the 
bifurcation parameter~$\mu $ by $\zeta =\root4\of\mu $, so that the norm of 
the $\zeta ^{4n+2}$ term in the power series for~$U$ is also bounded by a 
constant times $K^{n}(n!)^{4}<K^{n}(4n!)$. We use the \emph{Borel 
transform}~$\widehat{U}$ of the formal solution~$U$: 
the $\zeta ^{4n+2}$ term in the 
power series for $\widehat{U}$ is the $\zeta ^{4n+2}$ term in the power 
series for~$U$ divided by $(4n+2)!$. With this definition, $\widehat{U}$ is 
an analytic function of $\zeta $ in the disk $|\zeta |<K^{-1/4}$, and for 
each $\zeta $ in this disk, $\widehat{U}$ is a $Q$-fold quasiperiodic 
function of $(x,y)$ in the space~${\mathcal{H}}_{s}$. Of course the new 
function $\widehat{U}$ does not satisfy the original PDE, but we prove that 
it satisfies a transformed PDE (Theorem~\ref{BorelTransfthm}). 
 
The next stage would be to invert the Borel transform: however, the usual 
inverse Borel transform is a line integral (related to the Laplace 
transform) taking $\zeta $ from $0$ to~$\infty $, and $\widehat{U}$ is only 
an analytic function of $\zeta $ for $\zeta $ in a disk. If the definition 
of $\widehat{U}$ could be extended to a line in the complex~$\zeta $ plane, 
the inverse Borel transform would provide a quasiperiodic solution of 
the~\hbox{PDE} -- \emph{this remains an open problem}. 
 
Since the full inverse Borel transform cannot be used, in section~\ref%
{sec:TruncatedLaplace}, we use a truncated integral to define~$\bar{U}(\nu )$. 
This involves integrating~$\zeta $ along a line segment inside the disk 
where $\widehat{U}$~is analytic, weighted by an exponential that decays 
rapidly as $\nu \rightarrow 0$. We show that $\bar{U}(\nu )$ and $U(\mu )$ 
have the same power series expansion when we set $\nu =\root4\of\mu $,  
but unlike $U$, $\bar{U}(\nu )$ is a $C^{\infty }$ function of~$\nu $ in a 
neighbourhood of~$0$, taking values in~${\mathcal{H}}_{s}$. In other words, $%
\bar{U}(\mu ^{1/4})$ is a $Q$-fold quasiperiodic function of $(x,y)$ for small 
enough~$\mu $. This function is not an exact solution of the 
Swift--Hohenberg PDE, but we show in Theorem~\ref{sh_exp_estim} that the 
residual, when $\bar{U}(\mu ^{1/4})$ is substituted into the PDE, is 
exponentially small as $\mu \rightarrow 0$. Finally, in the last 
section~\ref{sec:initialvalueproblem},
we show that by taking as initial data the above approximate solution, the 
time dependent solution $U(t)$ stays exponentially close to the approximate 
solution for a long time, of the order $O(1/\mu ^{1+1/4l})$. 
 
In conclusion, we have shown that, for any even $Q\geq 8$, the divergent 
power series~$U(\mu )$ generated by the naive application of weakly 
nonlinear theory can be used to find a smooth $Q$-fold quasiperiodic
function~$\bar{U}(\mu ^{1/4l})$ that shares the same asymptotic expansion
as~$U$, and that satisfies the PDE with an exponentially small error.
 
This technique does not prove the existence of a quasiperiodic solution of 
the~\hbox{PDE}. However, this is a first step towards an existence proof for 
quasiperiodic solutions of PDEs like~(\ref{eq:shtime}). In particular, we 
may hope to use $\bar{U}$ as a starting point for the Newton iteration 
process that would form part of an existence proof using the Nash--Moser 
theorem. As an aside, ordinary numerical Newton iteration succeeds in 
finding an approximate solution of the truncated PDE for values of~$\mu $ 
where the formal power series has already diverged, as in figure~\ref%
{fig:amplitude}. 
 
An analogous result may be proved for example in the Rayleigh--B\'{e}nard 
convection problem (see \cite{Iooss2009a}), using the fact that the 
dispersion equation possesses the same property as in the present model: at 
the critical value of the parameter there is a circle of critical 
wavevectors in the plane. The method might also extend to the case of the 
Faraday wave experiment by considering fixed points of a stroboscopic map. 
 
In the present work we consider quasilattices generated by regularly spaced 
wavevectors on the unit circle, and solutions invariant under $2\pi/Q$ 
rotations. It might be worth studying the case of solutions having less 
symmetry on the same quasilattice, or quasilattices (still dense in the 
plane) generated by wavevectors that are irregularly spaced. 
 
\bigskip 
 
\emph{Acknowledgments}: We are grateful for useful discussions with Sylvie 
Benzoni, W. Crawley-Boevey, Andr\'{e} Galligo, Ian Melbourne, Jonathan 
Partington, David Sauzin and Gene Wayne. We are also grateful to the Isaac 
Newton Institute for Mathematical Sciences, where some of this work was 
carried out. 
 
\section{Small divisors: Quasilattices and Diophantine bounds} 
 
\label{sec:Diophantine} 
 
Let $Q\in\mathbb{N}$ ($Q\geq8$) be the order of a quasipattern and define 
wavevectors  
\begin{equation*} 
\mathbf{k}_{j}= \left( \cos \left( 2\pi\frac{j-1}{Q}\right), \sin \left( 2\pi 
\frac{j-1}{Q}\right) \right), \qquad j=1,2,\dots,Q 
\end{equation*} 
(see figure~\ref{fig:quasilattice}a). We define the \emph{quasilattice}~$%
\Gamma\subset\mathbb{R}^{2}$ to be the set of points spanned by integer 
combinations $\mathbf{k}_{\mathbf{m}}$ of the form  
\begin{equation} 
\mathbf{k}_{\mathbf{m}}=\sum_{j=1}^{Q}m_{j}\mathbf{k}_{j}, \qquad \text{where}
 \qquad \mathbf{m}=(m_{1},m_{2},\dots,m_{Q})\in\mathbb{N}^{Q}. 
\label{eq:bfkm} 
\end{equation} 
The set $\Gamma$ is dense in $\mathbb{R}^{2}$. 
 
We are interested in real functions $U(\mathbf{x})$ that are linear 
combinations of Fourier modes $e^{i\mathbf{k}\cdot\mathbf{x}}$, with $%
\mathbf{x}\in\mathbb{R}^{2}$ and $\mathbf{k}\in\Gamma$. If $U(\mathbf{x})$ 
is to be a real function, we need $Q$ to be even, with $\mathbf{k}_{j}$ and $%
-\mathbf{k}_{j}$ in $\Gamma$, hence the quasilattice~$\Gamma$ is symmetric 
with respect to the origin. 
 
\begin{figure}[tp] 
\hbox to \hsize{\hfil  
 \hbox to 0.3\hsize{\hfil (a)\hfil}\hfil              
 \hbox to 0.3\hsize{\hfil (b)\hfil}\hfil  
 \hbox to 0.3\hsize{\hfil (c)\hfil}\hfil} \vspace{0.5ex}  
\hbox to \hsize{\hfil                                           
  \mbox{\includegraphics[width=0.3\hsize]{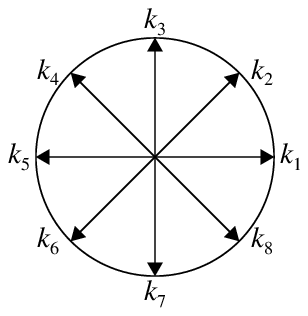}}\hfil  
  \mbox{\includegraphics[width=0.3\hsize]{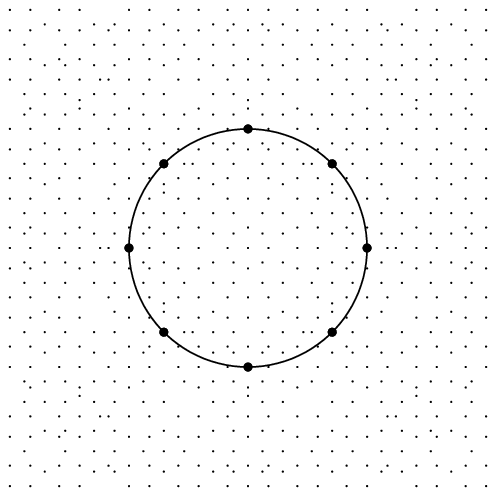}}\hfil  
  \mbox{\includegraphics[width=0.3\hsize]{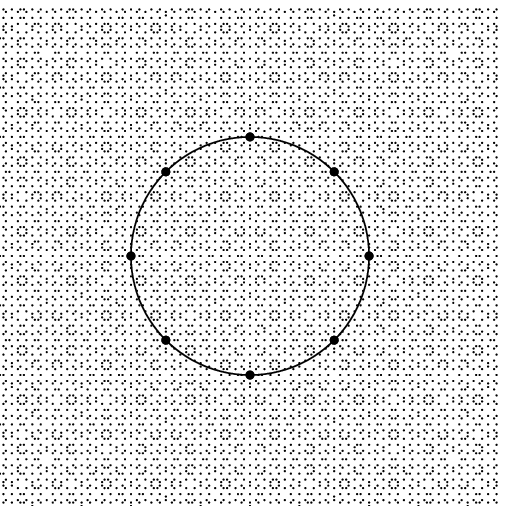}}  
\hfil} 
\caption{Example quasilattice with~$Q=8$, after~\protect\cite{Rucklidge2003}.
(a) The $8$ wavevectors with $|\mathbf{k}|=1$ that form the basis of the 
quasilattice. (b,c) The truncated quasilattices $\Gamma_9$ and~$\Gamma_{27}$. 
The small dots mark the positions of combinations of up to~$9$ or $27$ of 
the $8$~basis vectors on the unit circle. Note how the density of points 
increases with~$N_{\mathbf{k}}$.} 
\label{fig:quasilattice} 
\end{figure} 
 
In the calculations that follow, we will require Diophantine bounds on the 
magnitude of the small divisors. We see below that the small divisors are $%
\left\vert |\mathbf{k}|^{2}-1\right\vert $, for $\mathbf{k}\in \Gamma $. To 
compute the required lower bound, we start with  
\begin{equation*} 
|\mathbf{k}_{\mathbf{m}}|^{2}=\sum_{1\leq j_{1}<j_{2}\leq 
Q}2m_{j_{1}}m_{j_{2}}\cos (j_{1}-j_{2})\theta _{0}+\sum_{1\leq j\leq 
Q}m_{j}^{2}, 
\end{equation*}
where $\theta _{0}=2\pi /Q$. Let us define
\begin{equation*} 
\omega =2\cos \theta _{0}
\end{equation*}
We now show how $\left\vert |\mathbf{k}_{\mathbf{m}}|^{2}-1\right\vert $ can be expressed as
a polynomial in~$\omega$.

First, we can express $2\cos p\theta _{0}$ as a polynomial in $\omega $, for $%
1\leq p\leq Q-1$:
\begin{equation*} 
2\cos p\theta _{0}=\omega ^{p}-p\omega ^{p-2}+\frac{p(p-3)}{2}\omega 
^{p-4}\dots 
\end{equation*}
with integer coefficients which only depend on~$Q$ (easy proof by 
induction), and the leading coefficient being 1, and since $\cos 
(p+Q/2)\theta _{0}=-\cos \theta _{0}$, this leads to
\begin{equation} 
|\mathbf{k}_{\mathbf{m}}|^{2}=\sum_{0\leq r\leq Q/2-1}q_{r}^{\prime }\omega 
^{r},\text{ \ }q_{r}^{\prime }\in  
\mathbb{Z} 
,r=0,1,\dots,Q/2-1,  \label{k_m^2} 
\end{equation}
where the integers $q_{r}^{\prime }$ are quadratic forms of $\mathbf{m}$. 

Next, we use the property that $\omega $ is an \emph{algebraic integer}, since
it is the sum of two algebraic integers $e^{i\theta_{0}}+ e^{(Q-1)i\theta_{0}}$. 
More precisely, $\omega $~is a root of the (minimal) polynomial $P(\omega )$ with 
integer coefficients, with leading coefficient equal to~$1$, and which is of 
degree $\varphi (Q)/2:=l+1$, where $\varphi (Q)$ is Euler's Totient 
function~\cite{Barbeau1989}, the number of positive integers $j<Q$ such that $j$ 
and~$Q$ are relatively prime. 
For example, $\varphi (14)=6$ since the $6$ numbers  
$1$, $3$, $5$, $9$, $11$ and $13$ have no factors in common with $14$, and 
so $l+1=3$ in the case $Q=14$. In the cases $Q=8$, $10$ and $12$, the 
irrational numbers $\omega =2\cos \theta _{0}$ are $\sqrt{2}$, $\frac{1+%
\sqrt{5}}{2}$ and $\sqrt{3}$: these are quadratic algebraic numbers ($l+1=2$), 
while for $Q=14$, $\omega $ is cubic. 
 
Finally, dividing (\ref{k_m^2}) by $P(\omega )$ we obtain a remainder 
of degree $l$ such that  
\begin{equation} 
|\mathbf{k}_{\mathbf{m}}|^{2}-1=q_{0}+\omega q_{1}+\dots +\omega ^{l}q_{l} 
\label{k^2-1} 
\end{equation}
where $q_{0}+1$ and $q_{j}$, $j=1,\dots ,l$ are integer-valued 
\emph{quadratic forms} of $\mathbf{m}$. 
 
Define $|\mathbf{m}|=\sum_{j}m_{j}$, then, for a given wavevector $\mathbf{k}%
\in \Gamma $, we define the order $N_{\mathbf{k}}$ of $\mathbf{k}$ by  
\begin{equation} 
N_{\mathbf{k}}=\min \{|\mathbf{m}|;\mathbf{k}=\mathbf{k}_{\mathbf{m}},%
\mathbf{k}_{\mathbf{m}}\in \Gamma \}.  \label{defN_k} 
\end{equation}
The reason for this is that, for a given $\mathbf{k}$, there is an infinite 
set of $\mathbf{m}$'s satisfying $\mathbf{k}=\mathbf{k}_{\mathbf{m}}$. For 
example, we could increase $m_{j}$ and $m_{j+Q/2}$ by $1$: this increases $|%
\mathbf{m}|$ by $2$ but does not change $\mathbf{k}_{\mathbf{m}}$. Whenever 
solutions are computed numerically, it is necessary to use only a finite 
number of Fourier modes, so we define the \emph{truncated quasilattice}~$%
\Gamma _{N}$ to be:  
\begin{equation} 
\Gamma _{N}=\left\{ \mathbf{k}\in \Gamma :N_{\mathbf{k}}\leq N\right\} . 
\label{truncated_quasilattice} 
\end{equation}
Figure~\ref{fig:quasilattice}(b,c) shows the truncated quasilattices $\Gamma 
_{9}$ and $\Gamma _{27}$ in the case~$Q=8$. For example, we have in the case  
$Q=8$:  
\begin{eqnarray} 
\left\vert \mathbf{k}_{\mathbf{m}}\right\vert ^{2} &=&\sum_{j=1}^{4}{%
m_{j}^{\prime }}^{2}+\sqrt{2}\left( m_{1}^{\prime }m_{2}^{\prime 
}+m_{2}^{\prime }m_{3}^{\prime }+m_{3}^{\prime }m_{4}^{\prime 
}-m_{4}^{\prime }m_{1}^{\prime }\right) ,  \label{|k|^2} \\ 
N_{\mathbf{k}} &=&\sum_{j=1}^{4}|{m_{j}^{\prime }}|  \label{Nk} 
\end{eqnarray}
where $m_{j}^{\prime }=m_{j}-m_{j+Q/2}$. More generally we have  
\begin{equation*} 
N_{\mathbf{k}}\leq \sum_{j=1}^{Q/2}|{m_{j}^{\prime }}|. 
\end{equation*}
The above inequality can occur strictly (for example) in the case $Q=12$, 
because only $4$ of the $12$ vectors $\mathbf{k}_{j}$ are rationally 
independent in this case. More generally only $\varphi (Q)$ vectors $\mathbf{%
k}_{j}$ are rationally independent~\cite{Washington1997}. 
 
Now, the quantity in (\ref{k^2-1}), $|q_{0}+\omega q_{1}+\dots +\omega
^{l}q_{l}|$, may be as small as we want for good choices of large integers
$q_{j}$, and we need to have a lower bound when this is different from~$0$.
 
In~\cite{Rucklidge2003}, it was proved that in the cases $Q=8$, $10$ and $12$, 
there is a constant $c>0$ such that  
\begin{equation} 
\left||\mathbf{k}|^{2}-1\right|\geq\frac{c}{N_{\mathbf{k}}^2}, \qquad 
\text{for any $\mathbf{k}\in\Gamma$ with $|\mathbf{k}|\neq 1$.}  \label{dioph0} 
\end{equation} 
The proof relies on the fact that for quadratic algebraic numbers, there 
exists $C>0$ such that  
\begin{equation*} 
|p-\omega q|\geq \frac{C}{q} 
\end{equation*} 
holds for any $(p,q)\in\mathbb{Z}^{2}$, $q\neq 0$~\cite{Hardy1960}. Now 
using the fact that $q$ is quadratic in $\mathbf{m}$ (see~(\ref{|k|^2})) we 
have  
\begin{equation} 
q\leq QN_{\mathbf{k}}^{2}  \label{ineg_q} 
\end{equation} 
from which (\ref{dioph0}) can be deduced. 
 
The Diophantine bound~(\ref{dioph0}) may be extended to any even $Q\geq 8$, 
and there exists $c>0$ depending only on $Q$, such that for any $\mathbf{k}%
\in \Gamma $, with $|\mathbf{k}|\neq 1$, we have  
\begin{equation} 
\left\vert |\mathbf{k}|^{2}-1\right\vert \geq \frac{c}{N_{\mathbf{k}}^{2l}}. 
\label{dioph1} 
\end{equation}
To show this, we use the following known result (see \cite{Cohen1993}) 
proved in Appendix~\ref{app:dioph_estimate}: 
 
\begin{lemma} 
\label{dioph_estimate} Let $\omega $ be an algebraic number of order $l+1$, 
that is, a solution of $P(\omega )=0$ where $P$ is a polynomial of degree $%
l+1$ with integer coefficients, that is irreducible on $\mathbb{Q}$. 
Then, there exists a constant $C>0$ such that 
for any $\mathbf{q}=(q_{0},q_{1},\dots ,q_{l})\in \mathbb{Z}^{l+1}\backslash 
\{0\}$, the following estimate  
\begin{equation} 
|q_{0}+q_{1}\omega +q_{2}\omega ^{2}+\dots +q_{l}\omega ^{l}|\geq \frac{C}{|%
\mathbf{q}|^{l}}  \label{dioph_omega} 
\end{equation}
holds, where $|\mathbf{q}|=\sum_{0\leq j\leq l}|q_{j}|$. 
\end{lemma} 
 
In the general case, by choosing $\mathbf{m}$ such that $|\mathbf{m}|=N_{%
\mathbf{k}_{\mathbf{m}}}$, the estimate (\ref{ineg_q}) is replaced by
\begin{equation*} 
|\mathbf{q}|\leq c(Q)N_{\mathbf{k}}^{2} 
\end{equation*}
where $c(Q)$ depends only on $Q$. Then estimate (\ref{dioph1}) is satisfied 
by taking  
\begin{equation*} 
c=\frac{C}{[c(Q)]^{l}}. 
\end{equation*} 
 
It remains to show that $|\mathbf{k}|^{2}\neq1$ for all $\mathbf{k}\in\Gamma$,
apart from $\mathbf{k}=\mathbf{k}_1,\dots,\mathbf{k}_Q$. This is solved by
denoting $\zeta =e^{i\theta _{0}}$, the $Q$th primitive root of unity, and 
relating $\mathbf{k}_{j+1}$ to $\zeta^j$, and $\mathbf{k}_{\mathbf{m}}$ to 
$\sum_{j=0}^{Q-1}m_{j+1}\zeta ^{j}$. We then use the
Kronecker--Weber theorem which says that ``every abelian extension of $%
\mathbb{Q} 
$ is cyclotomic''~\cite{Washington1997}. 
This implies that if  
$\sum_{j}m_{j}\zeta ^{j}$ (which is an \emph{algebraic integer}) has 
modulus~$1$, then it is \emph{necessarily a root of unity}. 
Knowing that the dimension of the $%
\mathbb{Q} 
$-vector space spanned by the $\zeta ^{j}$ is $\varphi (Q)$, 
this implies that this root of 
unity is one of the $\zeta ^{j}$, $j=1,\dots,Q$.

\section{Formal power series computation} 
 
\label{sec:formal} 
 
Let us consider the \emph{steady Swift--Hohenberg equation}  
\begin{equation} 
(1+\Delta )^{2}U-\mu U+U^{3}=0  \label{eq:sh} 
\end{equation}
where we look for a $Q$-fold quasiperiodic function $U$ of $\mathbf{x}\in  
\mathbb{R}^{2}$, defined by Fourier coefficients $U_{\mathbf{k}}$ on a 
quasilattice~$\Gamma $ as defined above. Let us rewrite (\ref{eq:sh}) in the 
form  
\begin{equation} 
\mathcal{L}_{0}U=\mu U-U^{3}  \label{basicEqu} 
\end{equation}
where  
\begin{equation*} 
\mathcal{L}_{0}=(1+\Delta )^{2}. 
\end{equation*}
We write formally  
\begin{equation*} 
U(\mathbf{x})=\sum_{\mathbf{k}\in \Gamma }U_{\mathbf{k}}e^{i\mathbf{k}\cdot  
\mathbf{x}}, 
\end{equation*}
the meaning of this sum being given in section~\ref{sec:spaces}. We seek a 
solution of (\ref{eq:sh}), bifurcating from the origin when $\mu =0$, that 
is \emph{invariant under rotations by} $2\pi /Q$. First we look for a formal 
solution in the form of a power series of an amplitude. More precisely we 
look for the series  
\begin{equation} 
U(\mathbf{x},\mu )=\sqrt{\frac{\mu }{\beta }}\sum_{n\geq 0}\mu ^{n}U^{(n)}(%
\mathbf{x})  \label{Expan_U} 
\end{equation}
as a formal solution of (\ref{eq:sh}), where all factors $U^{(n)}$ are 
invariant under rotations by~$2\pi /Q$ of the plane. The coefficient $\beta $ 
will be given by fixing $U_{\mathbf{k}_{1}}^{(0)}$. 
 
At order $\mathcal{O}(\sqrt{|\mu |})$ in (\ref{eq:sh}) we have  
\begin{equation} 
\mathcal{L}_{0}U^{(0)}=0  \label{order1} 
\end{equation}
and we choose the solution  
\begin{equation} 
U^{(0)}=\sum_{j=1}^{Q}e^{i\mathbf{k}_{j}\cdot \mathbf{x}},  \label{U1} 
\end{equation}
which is invariant under rotations by $2\pi /Q$ and defined up to a factor 
which we take equal to~$1$. 

In writing $U^{(0)}$ in this way, we have made use of the fact that the only
solutions $\mathbf{k}\in\Gamma$ of $|\mathbf{k}|=1$ are
$\mathbf{k}=\mathbf{k}_1,\dots,\mathbf{k}_Q$ (see discussion at the end of
section~\ref{sec:Diophantine}). This implies that the kernel of
$\mathcal{L}_{0}$ is only one-dimensional if restricted to functions invariant
under rotations by~$2\pi/Q$, this kernel being spanned by~$U^{(0)}$.
 
At order $\mathcal{O}(|\mu |^{3/2})$ we have  
\begin{equation} 
\mathcal{L}_{0}U^{(1)}=U^{(0)}-\beta ^{-1}(U^{(0)})^{3}.  \label{def_beta} 
\end{equation}
We need to impose a solvability condition, namely that the coefficients of $%
e^{i\mathbf{k}_{j}\cdot \mathbf{x}}$, for $j=1,\dots ,Q$ on the right hand 
side of this equation must be zero. Because of the invariance under 
rotations by $2\pi /Q$, it is sufficient to cancel the coefficient of $e^{i%
\mathbf{k}_{1}\cdot \mathbf{x}}$. This yields  
\begin{equation} 
\beta =3(Q-1)>0,  \label{mu2} 
\end{equation}
and $U^{(1)}$ is known up to an element $\beta ^{(1)}U^{(0)}$ in $\ker  
\mathcal{L}_{0}$, which is determined at the next step:  
\begin{eqnarray} 
U^{(1)} &=&\widetilde{U}^{(1)}+\beta ^{(1)}U^{(0)},\qquad \widetilde{U}%
^{(1)}=\sum_{\mathbf{k}\in \Gamma ,N_{\mathbf{k}}=3}\alpha _{\mathbf{k}}e^{i%
\mathbf{k}\cdot \mathbf{x}},  \label{U^1} \\ 
\alpha _{3\mathbf{k}_{j}} &=&-1/64,\qquad \alpha _{2\mathbf{k}_{j}+\mathbf{k}%
_{l}}=-\frac{3}{(1-|2\mathbf{k}_{j}+\mathbf{k}_{l}|^{2})^{2}},\qquad \mathbf{%
k}_{j}+\mathbf{k}_{l}\neq 0,  \notag \\ 
\alpha _{\mathbf{k}_{j}+\mathbf{k}_{l}+\mathbf{k}_{r}} &=&-\frac{6}{(1-|%
\mathbf{k}_{j}+\mathbf{k}_{l}+\mathbf{k}_{r}|^{2})^{2}},\qquad j\neq l\neq 
r\neq j,  \notag \\ 
\mathbf{k}_{j}+\mathbf{k}_{l} &\neq &0,\qquad \mathbf{k}_{j}+\mathbf{k}%
_{r}\neq 0,\qquad \mathbf{k}_{r}+\mathbf{k}_{l}\neq 0,  \notag 
\end{eqnarray}
where $\widetilde{U}^{(1)}$ has no component on $e^{i\mathbf{k}_{j}\cdot  
\mathbf{x}}$. 
 
Order $|\mu |^{n+1/2}$ in (\ref{basicEqu}) leads for $n\geq 2$ to  
\begin{equation} 
\mathcal{L}_{0}U^{(n)}=U^{(n-1)}-\beta ^{-1}\sum_{\substack{ k+l+r=n-1,  \\ %
k,l,r\geq 0}}U^{(k)}U^{(l)}U^{(r)}.  \label{identifOrder_n} 
\end{equation}
For $n=2$, the solvability condition on the right hand side gives $\beta 
^{(1)}$, and $U^{(2)}$ is then determined up to $\beta ^{(2)}U^{(0)}$. 
Indeed we obtain on the right hand side  
\begin{eqnarray} 
U^{(1)}-\frac{3}{\beta }U^{(1)}U^{(0)2} &=&\widetilde{U}^{(1)}+\beta 
^{(1)}U^{(0)}-\frac{3}{\beta }\beta ^{(1)}U^{(0)3}-\frac{3}{\beta }%
\widetilde{U}^{(1)}U^{(0)2}  \notag \\ 
&=&-2\beta ^{(1)}U^{(0)}+\widetilde{U}^{(1)}-\frac{3}{\beta }\widetilde{U}%
^{(1)}U^{(0)2}-\frac{3}{\beta }\mathcal{L}_{0}\widetilde{U}^{(1)}, 
\label{beta_1} 
\end{eqnarray}
where we used the fact that the component of $U^{(0)3}$ on $e^{i\mathbf{k}%
_{1}\cdot \mathbf{x}}$ is $\beta $ (see (\ref{def_beta})). Hence $2\beta 
^{(1)}$ is the coefficient of $e^{i\mathbf{k}_{1}\cdot \mathbf{x}}$ in $%
-3\beta ^{-1}\widetilde{U}^{(1)}U^{(0)2}$, and since all coefficients of $%
\widetilde{U}^{(1)}$ are negative, we find $\beta ^{(1)}>0$. We obtain in 
the same way the coefficients $\beta ^{(n-1)}U^{(0)}$ of $U^{(n-1)}$ in 
using the solvability condition on the right hand side of~(\ref%
{identifOrder_n}). 
 
\begin{figure}[t] 
\begin{center} 
\includegraphics[width=0.7\hsize]{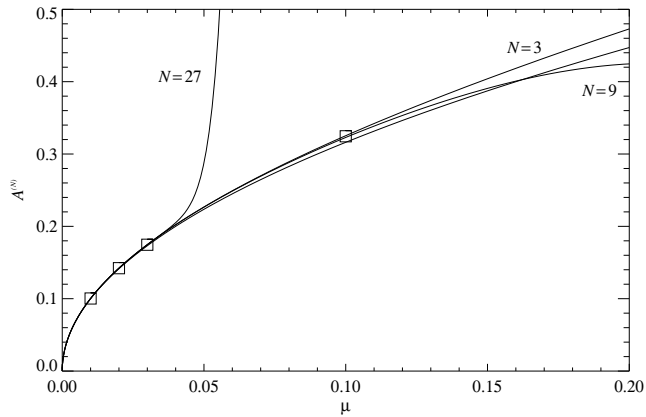} 
\end{center} 
\caption{Amplitude~$A^{(N)}$ of the quasipattern, as a function of~$\protect%
\mu$ and of~$N$, with $Q=8$, $N=1$, $3$, $9$ and~$27$, and scaled so that $%
A^{(1)}=\protect\sqrt{\protect\mu}$. Increasing the order of the truncation 
leads to divergence for smaller values of~$\protect\mu$. The squares 
represent amplitudes computed by solving the PDE by Newton iteration, 
truncated to the quasilattice~$\Gamma_{27}$ ($N_{\mathbf{k}}\le27$) and 
restricted to wavevectors with $|\mathbf{k}|\le\protect\sqrt{5}$. Note that 
for $\protect\mu=0.1$, the Newton iteration succeeds in finding an 
equilibrium solution of the PDE, while the formal power series has diverged. 
The spatial form of the solution with $\protect\mu=0.1$ is shown in 
figure~\protect\ref{fig:quasipattern}.} 
\label{fig:amplitude} 
\end{figure} 
 
\paragraph{Small divisor problem.} 
 
It is clear that we can continue to compute this expansion as far as we 
wish, where at each step we use the formal inverse of~$\mathcal{L}_{0}$ on 
the complement of the kernel. However, applying $\mathcal{L}_{0}^{-1}$ to $%
e^{i\mathbf{k}\cdot \mathbf{x}}$ introduces a factor  
\begin{equation*} 
\frac{1}{(1-|\mathbf{k}|^{2})^{2}}, 
\end{equation*}
which may be very large for combinations $\mathbf{k}=\mathbf{k}_{\mathbf{m}}$ 
with large $\mathbf{m}$, since points $\mathbf{k}_{\mathbf{m}}$ of the 
quasilattice $\Gamma $ sit as close as we want to the unit circle. This is a  
\emph{small divisor problem} and computations indicate that the series (\ref%
{Expan_U}) seems to diverge numerically~\cite{Rucklidge2003}. We illustrate 
this in figure~\ref{fig:amplitude}, plotting the amplitude~$A^{(N)}$ 
against~$\mu $, where  
\begin{equation*} 
A^{(N)}=||\mathbf{P}_{0}\sqrt{\frac{\mu }{\beta }}\sum_{n=0}^{(N-1)/2}\mu 
^{n}U^{(n)}||_{s}=\sqrt{\frac{\mu }{\beta }}\left( \sum_{n=0}^{(N-1)/2}\mu 
^{n}\beta ^{(n)}\right) ||U^{(0)}||_{s}, 
\end{equation*}
and the norm $||\cdot ||_{s}$ and the projection operator~$\mathbf{P}_{0}$ 
are defined below: $A^{(N)}$~is essentially the magnitude of the coefficient 
of~$e^{i\mathbf{k}_{1}\cdot \mathbf{x}}$ as a function of~$\mu $ and of~$N$, 
the maximum order of wavevectors included in the truncated power series. 
 
However, we prove in section~\ref{sec:Gevrey} that in all cases we can 
control the divergence of the terms of the series~(\ref{Expan_U}), 
and obtain a Gevrey estimate $||U^{(n)}||_{s}\leq \gamma K^{n}(n!)^{4l}$, 
where the norm $||\cdot ||_{s}$ is defined below. 
 
\begin{remark} 
For $Q=4$ or $6$, there is no small divisor problem since $\Gamma $ is a 
periodic lattice, and the only points in $\Gamma $ that lie in a small 
neighborhood of the unit circle are $\{\mathbf{k}_{j};j=1,\dots,Q\}$. 
\end{remark} 
 
\section{Function spaces} 
 
\label{sec:spaces} 
 
We characterise the functions of interest by their Fourier coefficients on 
the quasilattice $\Gamma $ generated by the $Q$ unit vectors 
$\mathbf{k}_{j}$:  
\begin{equation*} 
U(\mathbf{x})=\sum_{\mathbf{k}\in \Gamma }U_{\mathbf{k}}e^{i\mathbf{k}\cdot  
\mathbf{x}} 
\end{equation*}
Recall that for each $\mathbf{k}\in \Gamma $, there exists a vector $\mathbf{%
m}\in \mathbb{N}^{Q}$ such that $\mathbf{k}=\mathbf{k}_{\mathbf{m}%
}=\sum_{j=1}^{Q}m_{j}\mathbf{k}_{j}$ and we can choose $\mathbf{m}$ such 
that $|\mathbf{m}|=N_{\mathbf{k}}$ as defined in (\ref{defN_k}). We have the 
following properties, proved in Appendix~\ref{app:lemN_k}: 
 
\begin{lemma} 
\label{lem:N_k} 
 
(i) We have the following inequalities:  
\begin{equation} 
N_{\mathbf{k}+\mathbf{k}^{\prime }}\leq N_{\mathbf{k}}+N_{\mathbf{k}^{\prime 
}},\quad N_{-\mathbf{k}}=N_{\mathbf{k}},  \label{triangular} 
\end{equation}
\begin{equation} 
|\mathbf{k}|\leq N_{\mathbf{k}}.  \label{compare_k_and_N_k} 
\end{equation} 
 
(ii) We have the following estimate of the numbers of vectors $\mathbf{k}$ 
having a given $N_{\mathbf{k}}$:  
\begin{equation} 
\mathop{\rm card}\{\mathbf{k}:N_{\mathbf{k}}=N\}\leq c_{1}(Q)N^{Q/2-1} 
\label{card_k} 
\end{equation}
where $c_{1}(Q)$ only depends on $Q$. 
\end{lemma} 
 
Define now the space of functions  
\begin{equation} 
{\mathcal{H}}_{s}=\left\{U=\sum_{\mathbf{k}\in\Gamma } U_{\mathbf{k}}e^{i%
\mathbf{k}\cdot\mathbf{x}}: \, ||U||_{s}^{2}=\sum_{\mathbf{k}\in\Gamma} 
(1+N_{\mathbf{k}}{}^{2})^{s} |U_{\mathbf{k}}|^{2}<\infty \right\} , 
\label{defHs} 
\end{equation} 
which becomes a Hilbert space with the scalar product  
\begin{equation} 
\langle W,V\rangle_{s} = \sum_{\mathbf{k}\in\Gamma} (1+N_{\mathbf{k}%
}{}^{2})^{s} W_{\mathbf{k}}\overline{V}_{\mathbf{k}}.  \label{scalarproduct} 
\end{equation} 
In the sequel we need the following lemma, proved in Appendix~\ref%
{app:algebra}: 
 
\begin{lemma} 
\label{algebra}The space ${\mathcal{H}}_{s}$ is a Banach algebra for 
$s>Q/4$. In particular there exists $c_{s}>0$ such that  
\begin{equation} 
||UV||_s\leq c_{s} ||U||_s ||V||_s.  \label{algeb} 
\end{equation} 
For $\ell\geq0$ and $s>\ell+Q/4$, ${\mathcal{H}}_{s}$ is continuously 
embedded into $\mathcal{C}^{\ell }$. 
\end{lemma} 
 
From now on, all inner products are $s$ unless otherwise stated, so that we 
can remove the $s$ subscripts throughout in scalar products. 
 
We will also use the orthogonal projection on $\ker\mathcal{L}_{0}$: for any  
$U\in{\mathcal{H}}_{s}$, let  
\begin{equation*} 
\mathbf{P}_{0}U=\sum_{j=1,\dots,Q} U_{\mathbf{k}_{j}}e^{i\mathbf{k}_{j}\cdot%
\mathbf{x}}, 
\end{equation*} 
and we denote by $Q_{0}$ the orthogonal projection:  
\begin{equation*} 
\mathbf{Q}_{0}=\mathbb{I}-\mathbf{P}_{0}, 
\end{equation*} 
which consists in suppressing the Fourier components of $e^{i\mathbf{k}%
_{j}\cdot\mathbf{x}}$, $j=1,\dots,Q$. The norm of the linear operator $%
\mathbf{Q}_{0}$ is $1$ in all spaces ${\mathcal{H}}_{s}$. 
 
\section{Gevrey estimates} 
 
\label{sec:Gevrey} 
 
In this section we prove rigorously a Gevrey estimate of $%
U^{(n)}$ in~(\ref{Expan_U}). The estimate for $Q=8$, $10$ and $12$ ($l=1$) 
was announced in~\cite{Rucklidge2003}. Recall that a formal power series $%
\sum_{n=0}^{\infty }u_{n}\zeta ^{n}$ is \emph{Gevrey-}$k$~\cite{Gevrey1918}, 
where $k$~is a positive integer, if there are constants $\delta >0$ and $%
K>0$ such that  
\begin{equation} 
|u_{n}|\leq \delta K^{n}(n!)^{k}\qquad \forall n\geq 0.  \label{defn:Gevery} 
\end{equation} 
 
\begin{theorem} 
\label{gevreythm} For any even $Q\geq 8$, assume that $s>Q/4$. Then there 
exist positive numbers $K(Q,c,s)$ and $\delta (Q,s)$ such that there exists 
a unique formal solution $U(\mu )$ of~(\ref{eq:sh}), under the form of a 
power series in $\mu ^{1/2}$, all factors $U^{(n)}$ belonging to ${\mathcal{H%
}}_{s}$, and which satisfies  
\begin{eqnarray} 
U &=&\sqrt{\frac{\mu }{\beta }}\sum_{n\geq 0}\mu ^{n}U^{(n)}, 
\label{formal_exp} \\ 
U^{(n)} &=&\beta ^{(n)}U^{(0)}+\widetilde{U}^{(n)},\qquad \langle \widetilde{%
U}^{(n)},e^{i\mathbf{k}_{j}\cdot \mathbf{x}}\rangle _{s}=0,\quad j=1,\dots 
,Q,  \notag \\ 
||\widetilde{U}^{(n)}||_{s} &\leq &\delta \frac{(Q-1)}{2^{s/2}c_{s}^{2}Q}%
K^{n}(n!)^{4l},\quad n\geq 1,  \notag \\ 
|\beta ^{(n)}| &\leq &\delta K^{n}(n!)^{4l},\quad n\geq 1.  \notag 
\end{eqnarray}
where $l=\frac{1}{2}\varphi (Q)-1$~is the integer defined in Lemma~\ref%
{dioph_estimate}. From the above inequalities, it follows that  
\begin{equation*} 
||U^{(n)}||_{s}\leq \gamma K^{n}(n!)^{4l},\quad n\geq 0, 
\end{equation*}
where $\gamma $ is related to $\delta $, $Q$ and~$s$ only. 
\end{theorem} 
 
\begin{remark} 
The above Theorem claims that the series $U$ in powers of $\sqrt{\mu }$ is 
Gevrey-$2l$ taking its values in ${\mathcal{H}}_{s}$. 
\end{remark} 
 
\begin{remark} 
In the cases when $Q=4$ or $6$, the pattern is periodic, and the above 
series may be built in the same way, leading to a series which is convergent 
for $\mu <\mu _{0}$ where $\mu _{0}>0$. This results simply, via the 
Lyapunov--Schmidt method, from the implicit function theorem in its analytic 
version. The values of $\mu_{0}$ for $Q=2$, $4$ or $6$ are estimated 
in~\cite{Rucklidge2003}.
\end{remark} 
 
\begin{proof} 
\ We choose $s>Q/4$ since Lemma~\ref{algebra} insures that ${\mathcal{H}}_{s} 
$ is then a Banach algebra. We notice that  
\begin{equation*} 
||e^{i\mathbf{k}_{j}\cdot \mathbf{x}}||_{s}=2^{s/2}, 
\end{equation*}
and  
\begin{equation} 
||U^{(0)}||_{s}=2^{s/2}\sqrt{Q}.  \label{estimU^0} 
\end{equation}
We also have $\beta ^{(0)}=1$ and $\widetilde{U}^{(0)}=0$. Now we notice 
from~(\ref{dioph1}) that for $|\mathbf{k}|\neq 1$ we have  
\begin{equation} 
\left\vert |\mathbf{k}|^{2}-1\right\vert ^{-2}\leq \frac{N_{\mathbf{k}}^{4l}%
}{c^{2}},  \label{dioph2} 
\end{equation}
which controls the unboundedness of the pseudo-inverse $\widetilde{\mathcal{L%
}}_{0}^{-1}$ (inverse of $\mathcal{L}_{0}$ restricted to the orthogonal 
complement of its kernel). Indeed $\widetilde{\mathcal{L}}_{0}^{-1}$ is 
bounded from ${\mathcal{H}}_{s}$ to $\mathcal{H}_{s-4l}$. 
 
\begin{remark} 
We may notice that the set of eigenvalues of $\mathcal{L}_{0}$ is dense in 
the positive real line, which constitutes the spectrum. Hence \emph{0 is not 
isolated in the spectrum of }$\mathcal{L}_{0}$. This explains why the 
pseudo-inverse of $\mathcal{L}_{0}$ on the complement of its kernel, is 
unbounded and satisfies (see (\ref{dioph2})):
\begin{equation*} 
||\widetilde{\mathcal{L}_{0}}^{-1}\mathbf{Q}_{0}U||_{s-4l}\leq \frac{1}{c^{2}%
}||U||_{s},\text{ for any }U\in \mathcal{H}_{s}. 
\end{equation*} 
\end{remark} 
 
The basic observation here is that the factor $U^{(n)}$ that multiplies $\mu ^{n}$ 
has a finite Fourier expansion in $e^{i\mathbf{k}\cdot \mathbf{x}}$, with $%
\mathbf{\ k}=\sum_{j=1}^{Q}m_{j}\mathbf{k}_{j}$, $\sum m_{j}\leq 2n+1$, 
hence $N_{\mathbf{k}}\leq 2n+1$. Since for $\widetilde{U}^{(1)}$ we have $|%
\mathbf{m}|=3$ in all $\mathbf{k}_{\mathbf{m}}$'s, equation~(\ref{def_beta}) 
leads to  
\begin{equation} 
||\widetilde{U}^{(1)}||_{s}\leq \frac{3^{4l}c_{s}^{2}2^{3s/2}Q^{3/2}}{%
c^{2}3(Q-1)}.  \label{estimU^1} 
\end{equation}
We set  
\begin{equation} 
U^{(n)}=\beta ^{(n)}U^{(0)}+\widetilde{U}^{(n)},\qquad \widetilde{U}^{(n)}=%
\mathbf{Q}_{0}U^{(n)},  \label{decomp_U} 
\end{equation}
and replacing this decomposition in~(\ref{identifOrder_n}) we obtain, by 
taking the scalar product with $e^{i\mathbf{k}_{1}\cdot \mathbf{x}}$  
\begin{equation*} 
\beta ^{(n-1)}2^{s}-\frac{1}{\beta }\left\langle 3U^{(n-1)}U^{(0)2},e^{i%
\mathbf{k}_{1}\cdot \mathbf{x}}\right\rangle -\frac{1}{\beta }\left\langle 
\sum_{\substack{ k+l+r=n-1, \\ 0\leq k,l,r\leq n-2}}%
U^{(k)}U^{(l)}U^{(r)},e^{i\mathbf{k}_{1}\cdot \mathbf{x}}\right\rangle =0, 
\end{equation*}
where we have used $\left\langle U^{(0)},e^{i\mathbf{k}_{1}\cdot \mathbf{x}%
}\right\rangle =||e^{i\mathbf{k}_{1}\cdot \mathbf{x}}||_{s}^{2}=2^{s}$. 
Next, we use  
\begin{eqnarray*} 
\langle 3U^{(n-1)}U^{(0)2},e^{i\mathbf{k}_{1}\cdot \mathbf{x}}\rangle  
&=&\beta ^{(n-1)}\langle 3U^{(0)3},e^{i\mathbf{k}_{1}\cdot \mathbf{x}%
}\rangle +\langle 3\widetilde{U}^{(n-1)}U^{(0)2},e^{i\mathbf{k}_{1}\cdot  
\mathbf{x}}\rangle  \\ 
&=&3\beta \beta ^{(n-1)}2^{s}+\langle 3\widetilde{U}^{(n-1)}U^{(0)2},e^{i%
\mathbf{k}_{1}\cdot \mathbf{x}}\rangle , 
\end{eqnarray*}
and we are led to solve with respect to $\beta ^{(n-1)},\widetilde{U}^{(n)}$ 
the following system for $n\geq 2$  
\begin{eqnarray} 
\mathcal{L}_{0}\widetilde{U}^{(n)} &=&\widetilde{U}^{(n-1)}-\beta ^{-1}%
\mathbf{Q}_{0}\sum_{\substack{ k+l+r=n-1, \\ k,l,r\geq 0}}%
U^{(k)}U^{(l)}U^{(r)},  \label{Utilde^n} \\ 
\beta ^{(n-1)} &=&\frac{-1}{2^{1+s}\beta }\left\langle 3\widetilde{U}%
^{(n-1)}U^{(0)2}+\sum_{\substack{ k+l+r=n-1, \\ 0\leq k,l,r\leq n-2}}%
U^{(k)}U^{(l)}U^{(r)},e^{i\mathbf{k}_{1}\cdot \mathbf{x}}\right\rangle . 
\label{beta^n-1} 
\end{eqnarray}
Now we make the following recurrence assumption: there exist positive 
constants $\gamma _{1}$, $\delta $ and $K$, depending on $Q$ and $s$, such 
that  
\begin{eqnarray} 
||\widetilde{U}^{(p)}||_{s} &\leq &\gamma _{1}K^{p}(p!)^{4l},\qquad 
p=0,1,\dots ,n-1,  \label{recurrence} \\ 
|\beta ^{(p)}| &\leq &\delta K^{p}(p!)^{4l},\qquad p=1,\dots ,n-2.  \notag 
\end{eqnarray}
These estimates hold for $\widetilde{U}^{(0)}=0$ and for $\widetilde{U}^{(1)} 
$ provided that $\gamma _{1}$ and $K$ satisfy  
\begin{equation} 
\frac{3^{4l}c_{s}^{2}2^{3s/2}Q^{3/2}}{c^{2}3(Q-1)}\leq \gamma _{1}K. 
\label{cond_gamma_k} 
\end{equation}
Putting these together results in  
\begin{equation*} 
||U^{(p)}||_{s}=||\beta ^{(p)}U^{(0)}+\widetilde{U}^{(p)}||_{s}\leq \left( 
2^{s/2}\delta \sqrt{Q}+\gamma _{1}\right) K^{p}(p!)^{4l}, 
\end{equation*}
or  
\begin{equation} 
||U^{(p)}||_{s}\leq \gamma K^{p}(p!)^{4l},\qquad \text{with}\qquad \gamma 
=2^{s/2}\delta \sqrt{Q}+\gamma _{1}.  \label{eq:U_gamma} 
\end{equation}
The resolution of (\ref{Utilde^n}) and (\ref{beta^n-1}) provides $\beta 
^{(n-1)}$ and $\widetilde{U}^{(n)}$, starting with $n=2$. A useful lemma is 
the following, proved in Appendix~\ref{app:factorials}. 
 
\begin{lemma} 
\label{factorials} The following estimates hold true for $l\geq 1$:  
\begin{eqnarray*} 
\Pi _{3,n} &=&\sum_{\substack{ k+l+r=n  \\ k,l,r\geq 0}}(k!l!r!)^{4l}\leq 
4(n!)^{4l},\quad n\geq 1 \\ 
\Pi _{3,n}^{\prime } &=&\sum_{\substack{ k+l+r=n  \\ 0\leq k,l,r\leq n-1}}%
(k!l!r!)^{4l}\leq 10((n-1)!)^{4l},\quad n\geq 2. 
\end{eqnarray*} 
\end{lemma}

Thanks to Lemma~\ref{factorials} and the estimate for $||U^{(p)}||_{s}$ in (%
\ref{eq:U_gamma}), we observe that  
\begin{equation*} 
\left\Vert \sum_{\substack{ k+l+r=n-1 \\ 0\leq k,l,r\leq n-2}}%
U^{(k)}U^{(l)}U^{(r)}\right\Vert _{s}\leq 10c_{s}^{2}\gamma 
^{3}K^{n-1}((n-2)!)^{4l}. 
\end{equation*}
From this it follows that  
\begin{equation*} 
|\beta ^{(n-1)}|\leq \frac{c_{s}^{2}}{2^{1+s/2}\beta }K^{n-1}((n-1)!)^{4l}%
\left\{ 3\gamma _{1}2^{s}Q+10\gamma ^{3}\right\} , 
\end{equation*}
and the recurrence assumption is realized if  
\begin{equation} 
\frac{c_{s}^{2}}{3(Q-1)2^{1+s/2}}\left\{ 3\gamma _{1}2^{s}Q+10\gamma 
^{3}\right\} \leq \delta   \label{cond_on_delta} 
\end{equation}
holds. Now we have, still by using Lemma~\ref{factorials}  
\begin{eqnarray} 
||\widetilde{U}^{(n)}||_{s} &\leq &\frac{(2n+1)^{4l}K^{n-1}((n-1)!)^{4l}}{%
c^{2}}\left\{ \gamma _{1}+\frac{4c_{s}^{2}}{\beta }\gamma ^{3}\right\}  
\label{eq:estimateUtilden} \\ 
&\leq &K^{n}(n!)^{4l}(2+\frac{1}{n})^{4l}\frac{1}{Kc^{2}}\left\{ \gamma _{1}+%
\frac{4c_{s}^{2}}{\beta }\gamma ^{3}\right\} .  \notag 
\end{eqnarray}
The factor $(2n+1)^{4l}/c^{2}$ here comes from the pseudo-inverse of $%
\mathcal{L}_{0}$ acting on functions containing modes of order up to $2n+1$. 
The recurrence assumption is realized if  
\begin{equation} 
\frac{3^{4l}}{c^{2}}\left\{ \gamma _{1}+\frac{4c_{s}^{2}}{\beta }\gamma 
^{3}\right\} \leq \gamma _{1}K.  \label{cond_on_K} 
\end{equation} 
 
We now must choose $\gamma _{1}$, $\delta $ and $K$ in such a way as to 
satisfy the three conditions (\ref{cond_gamma_k}), (\ref{cond_on_delta}) 
and~(\ref{cond_on_K}). Indeed, we may choose $\gamma _{1}$ such that
\begin{equation*} 
\gamma _{1}=\delta \frac{(Q-1)}{2^{s/2}c_{s}^{2}Q}, 
\end{equation*}
and replacing this value in (\ref{eq:U_gamma}) and (\ref{cond_on_delta}), 
then (\ref{cond_on_delta}) is satisfied as soon as
\begin{equation*} 
\delta ^{2}\leq \frac{3(Q-1)2^{s/2-1}}{5c_{s}^{2}}\left( 2^{s/2}\sqrt{Q}+%
\frac{Q-1}{2^{s/2}c_{s}^{2}Q}\right) ^{-3} 
\end{equation*}
holds. Then, choosing $K$ such that
\begin{equation*} 
K=\max \left\{ \frac{3^{4l}}{c^{2}}\left( 1+\frac{2^{s+1}c_{s}^{2}Q}{5(Q-1)}%
\right) ,\frac{1}{\delta }\frac{3^{4l-1}2^{2s}c_{s}^{4}Q^{5/2}}{%
c^{2}(Q-1)^{2}}\right\} , 
\end{equation*}
allows to satisfy (\ref{cond_gamma_k}) and~(\ref{cond_on_K}). 
 
We conclude that the bounds on $||\widetilde{U}^{(n)}||_{s}$ and $|\beta 
^{(n)}|$ in Theorem~\ref{gevreythm} hold, and that~(\ref{eq:U_gamma}), which 
holds for $0\leq p\leq n-1$, also holds for $p=n$, and so  
\begin{equation*} 
||U^{(n)}||_{s}\leq \gamma K^{n}(n!)^{4l},\quad n\geq 1. 
\end{equation*}
This ends the proof of Theorem~\ref{gevreythm}. 
\end{proof} 
 
\section{Borel transform of the formal solution} 
 
\label{sec:Borel} 
 
In this and subsequent sections, we consider the cases with $l=1$ ($Q=8$, $%
10 $ and~$12$) and set  
\begin{equation*} 
\sqrt{\mu }=\zeta ^{2}. 
\end{equation*} 
 
\begin{remark} 
\label{alpha_diff_2}In the general case, we should set $\zeta =\mu ^{1/4l}$. 
\end{remark} 
 
The formal expansion~(\ref{formal_exp}) becomes, after incorporating $\beta 
^{-1/2}$ into~$U^{(n)}$,  
\begin{equation} 
U=\zeta ^{2}\sum_{n\geq 0}\zeta ^{4n}U^{(n)},  \label{formal_exp_zeta} 
\end{equation} 
and we have the estimate  
\begin{equation*} 
||U^{(n)}||_{s}\leq \gamma K^{n}(n!)^{4}\leq \gamma K^{n}(4n!). 
\end{equation*} 
Thus the formal power series (\ref{formal_exp_zeta}) is a Gevrey-$1$ series 
in~$\zeta$. 
 
Let us now consider the new function $\zeta\mapsto\widehat{U}(\zeta)$, 
taking its values in~${\mathcal{H}}_{s}$, defined by  
\begin{equation*} 
\widehat{U}(\zeta)=\sum_{n\geq 0}\frac{\zeta^{4n+2}}{(4n+2)!}U^{(n)}. 
\end{equation*} 
Indeed, by construction, this function is analytic in the disc $%
|\zeta|<K_{1}^{-1}=K^{-1/4}$, with values in the Hilbert space~${\mathcal{H}}%
_{s}$ and invariant under rotations of angle $2\pi/Q$. The mapping $U\mapsto%
\widehat{U}$, where we divide the coefficient of~$\zeta^{n} $ by $n!$, is 
the \emph{Borel transform}~\cite{Borel1901} applied to the series~$U$. Since  
$U$ satisfies a Gevrey-$1$ estimate, the Borel transform~$\widehat{U}$ is 
analytic in a disc. 
 
We now need to show that this function~$\widehat{U}(\zeta )$ is solution of 
a certain partial differential equation. Let us recall a simple property of 
Gevrey-$1$ series. Consider two scalar Gevrey-$1$ series $u$ and $v$  
\begin{eqnarray*} 
u &=&\sum_{n\geq 1}u_{n}\zeta ^{n},\qquad v=\sum_{n\geq 1}v_{n}\zeta ^{n}, \\ 
|u_{n}| &\leq &c_{1}K_{1}^{n}n!,\qquad |v_{n}|\leq c_{2}K_{1}^{n}n!, 
\end{eqnarray*}
then we have  
\begin{eqnarray*} 
(uv)_{n} &=&\sum_{1\leq k\leq n-1}u_{k}v_{n-k}, \\ 
|(uv)_{n}| &\leq &c_{1}c_{2}K_{1}^{n}n!, 
\end{eqnarray*}
as this results from Appendix~\ref{app:factorials}, by using the following 
inequality for $n\geq 3$  
\begin{equation*} 
\frac{1}{(n-1)!}\sum_{1\leq k\leq n-1}k!(n-k)!\leq 1+2(\frac{1}{2}+\dots +%
\frac{1}{n-1})\leq n, 
\end{equation*}
which shows that in our case we can multiply two Gevrey-$1$ series with 
coefficients belonging to ${\mathcal{H}}_{s}$ (the factor $c_{1}c_{2}$ is 
then multiplied by $c_{s})$ and obtain a new Gevrey-$1$ series with 
coefficients in ${\mathcal{H}}_{s}$. It is then classical that we can write  
\begin{equation} 
\widehat{U^{3}}=\widehat{U}\ast _{G}\widehat{U}\ast _{G}\widehat{U} 
\label{convolU^3} 
\end{equation}
where the \emph{convolution product}, written as $\ast _{G}$, is well 
defined by  
\begin{equation*} 
(\hat{u}\ast _{G}\hat{v})(\zeta )=\sum_{n\geq 1}\sum_{1\leq k\leq n-1}\frac{%
u_{k}v_{n-k}}{n!}\zeta ^{n}, 
\end{equation*}
and satisfies  
\begin{equation*} 
(\hat{u}\ast _{G}\hat{v})=\widehat{(uv)}. 
\end{equation*}
This convolution product is easily extended for two functions $f(\zeta )$ 
and~$g(\zeta )$, analytic in the disc $|\zeta |<K_{1}^{-1}$, and with no 
zero order term, by  
\begin{equation} 
(f\ast g)(\zeta )=\sum_{n\geq 1}\sum_{1\leq k\leq n-1}f_{k}g_{n-k}\frac{%
k!(n-k)!}{n!}\zeta ^{n}.  \label{defconvol} 
\end{equation}
It is clear that for $f=\hat{u}$, and $g=\hat{v}$ we have  
\begin{equation*} 
f\ast g=(\hat{u}\ast _{G}\hat{v})=\widehat{(uv)}. 
\end{equation*}
Since we have (\ref{convolU^3}), it is clear from (\ref{identifOrder_n}) 
that we have  
\begin{equation*} 
(\widehat{(1+\Delta )^{2}U})(\mathbf{x},\zeta )=(1+\Delta )^{2}\widehat{U}(%
\mathbf{x},\zeta ). 
\end{equation*}
Now let us define a bounded linear operator $\mathcal{K}$ as follows: for 
any function $\zeta \mapsto V(\zeta )$ analytic in the disc $|\zeta 
|<K_{1}^{-1}$, taking values in ${\mathcal{H}}_{s}$, canceling for 
$\zeta=0$, and satisfying  
\begin{equation*} 
V(\zeta )=\sum_{n\geq 1}V_{n}\zeta ^{n},\text{ \ }||V_{n}||_{s}\leq 
cK_{1}^{n}, 
\end{equation*}
we define  
\begin{equation*} 
(\mathcal{K}V)(\zeta )=\sum_{n\geq 1}\frac{n!}{(n+4)!}\zeta ^{n+4}V_{n}. 
\end{equation*}
It is then clear for $V=\widehat{U}$ that  
\begin{equation*} 
(\mathcal{K}\widehat{U})(\zeta )=\sum_{n\geq 0}\frac{\zeta ^{4n+6}}{(4n+6)!}%
U^{(n)}=\widehat{(\zeta ^{4}U)}, 
\end{equation*}
and we see that  
\begin{equation*} 
\partial _{\zeta }^{4}(\mathcal{K}\widehat{U})=\widehat{U}. 
\end{equation*} 
 
We now claim the following: 
 
\begin{theorem} 
\label{BorelTransfthm} The Borel transform $\widehat{U}(\mathbf{x},\zeta )$ 
of the Gevrey solution found in Theorem~\ref{gevreythm} for $l=1$ is the 
unique solution, analytic in the disc $|\zeta |<K^{-1/4}$, cancelling for $%
\zeta =0$, and taking values in ${\mathcal{H}}_{s}$ invariant under 
rotations of angle $2\pi /Q$, of the equation  
\begin{equation} 
(1+\Delta )^{2}V-\mathcal{K}V+V\ast V\ast V=0.  \label{Borel_PDEqu} 
\end{equation} 
\end{theorem} 
 
\begin{proof} 
We assume $l=1$ in what follows. The changes needed for larger $l$'s are 
left to the reader. Let us look for a solution $V$ in the form  
\begin{equation*} 
V=\sum_{n\geq 1}\zeta ^{n}V_{n}, 
\end{equation*}
where $V_{n}\in {\mathcal{H}}_{s}$ is invariant under rotations of angle $%
2\pi /Q$. Then defining a formal series  
\begin{equation*} 
U=\sum_{n\geq 1}\zeta ^{n}U_{n},\qquad U_{n}=n!V_{n}, 
\end{equation*}
it is clear that $U$ satisfies formally  
\begin{equation*} 
(1+\Delta )^{2}U-\mathcal{\zeta }^{4}U+U^{3}=0, 
\end{equation*}
and by identifying powers of~$\zeta $:  
\begin{eqnarray*} 
\mathcal{L}_{0}U_{1} &=&0, \\ 
\mathcal{L}_{0}U_{2} &=&0, \\ 
\mathcal{L}_{0}U_{3}+U_{1}^{3} &=&0, 
\end{eqnarray*}
which leads to $U_{1}=0$ because of the last equation where the solvability 
condition cannot be satisfied. Then we have  
\begin{equation*} 
U_{1}=0,\qquad \mathcal{L}_{0}U_{j}=0,j=2,3,4,5, 
\end{equation*}
and  
\begin{equation*} 
\mathcal{L}_{0}U_{6}-U_{2}+(U_{2})^{3}=0. 
\end{equation*}
We observe that $U_{2}$ and $U_{6}$ satisfy the equations verified by $\beta 
^{-1/2}U^{(0)}$ and $\beta ^{-1/2}U^{(1)}$ (see~(\ref{def_beta})). This is 
indeed the only solution invariant under rotations of $2\pi /Q$. Hence  
\begin{eqnarray*} 
U_{2} &=&\beta ^{-1/2}U^{(0)}, \\ 
U_{6} &=&\beta ^{-1/2}U^{(1)}. 
\end{eqnarray*}
Now at order $\zeta ^{7}$ we get  
\begin{equation*} 
\mathcal{L}_{0}U_{7}-U_{3}+3U_{2}^{2}U_{3}=0 
\end{equation*}
and since $U_{3}=CU^{(0)}$, where $C$ is a constant, the solvability 
condition gives  
\begin{equation*} 
C=\frac{3C}{\beta }\langle U^{(0)3},e^{i\mathbf{k}_{1}\cdot \mathbf{x}%
}\rangle _{s}=3C 
\end{equation*}
hence $C=0$ and $U_{3}=0$. It is the same for $U_{4}=U_{5}=0$, and we obtain  
$\mathcal{L}_{0}U_{7}=\mathcal{L}_{0}U_{8}=\mathcal{L}_{0}U_{9}=0$. Then the 
computation of higher orders is exactly as the one for the computation of $%
U^{(n)}$, since the cubic term cancels if the sum of the 3 indices $p$ in $%
U_{p}$ is not $2\mod4$. Coming back to the definition of $U_{n}=n!V_{n}$, it 
is then clear that Theorem~\ref{BorelTransfthm} is proved. 
\end{proof} 
 
\section{Truncated Laplace transform} 
 
\label{sec:TruncatedLaplace} 
 
Let us take $K^{\prime }>K_{1}$ and define a linear mapping $U\mapsto\bar{U}$ 
in the set of Gevrey-$1$ series taking values in ${\mathcal{H}}_{s}$  
\begin{equation} 
\bar{U}(\nu )=\frac{1}{\nu }\int_{0}^{\frac{1}{K^{\prime }}}e^{-\frac{ \zeta  
}{\nu }}\widehat{U}(\zeta )\,d\zeta ,  \label{approxinvBorel} 
\end{equation} 
where $\widehat{U}(\zeta )$ is the Borel transform of $U$ as defined above, 
which is analytic in the disc $|\zeta |<1/K_{1}$. The function $\nu \mapsto  
\bar{U}(\nu )$ is a truncated Laplace transform of the Borel transform of~$U$.
 
\begin{remark} 
If $\widehat{U}(\zeta)$ could be shown to be analytic on a line in the 
complex $\zeta$ plane extending to~$\infty$, instead of just in a disk, then 
the Laplace transform in~(\ref{approxinvBorel}) would be the inverse Borel 
transform, and would provide a quasiperiodic solution of~(\ref{eq:sh}) in~${%
\mathcal{H}}_{s}$. 
\end{remark} 
 
It is clear that $\bar{U}(\nu )$ is a $\mathcal{C}^{\infty }$ function of $%
\nu $ in a neighborhood of 0, taking its values in ${\mathcal{H}}_{s}$, as 
this results from  
\begin{equation*} 
\bar{U}(\nu )=\int_{0}^{\frac{1}{K^{\prime }\nu }}e^{-z}\widehat{U}(\nu 
z)\,dz 
\end{equation*}
and from the dominated convergence theorem. Moreover $\bar{U}(\nu )$ and $%
U(\mu )$ have the same asymptotic expansion in powers on $\nu $, when we set  
$\mu =\nu ^{1/4}$, as this results from  
\begin{equation} 
\frac{1}{\nu }\int_{0}^{\frac{1}{K^{\prime }}}e^{-\frac{\zeta }{\nu }}\frac{%
\zeta ^{n}}{n!}\,d\zeta =\nu ^{n}-e^{-\frac{1}{K^{\prime }\nu }}\left( \frac{%
\nu ^{n}}{1}+\frac{\nu ^{n-1}}{K^{\prime }1!}+\dots +\frac{\nu }{K^{\prime 
}{}^{n-1}(n-1)!}+\frac{1}{K^{\prime }{}^{n}n!}\right) . 
\label{basic_identity} 
\end{equation}
It is also clear that in a little disc near the origin  
\begin{equation*} 
\widehat{\bar{U}}=\widehat{U}, 
\end{equation*}
but this does not imply that $\bar{U}=U$ since $U$ is not a function, being 
defined as a formal series of~$\nu ^{4}$, and an asymptotic expansion does 
not define a unique function. The real question is whether or not $\bar{U}$ 
is solution of (\ref{eq:sh}) in~${\mathcal{H}}_{s}$. 
 
By construction, we know that the Gevrey-$1$ expansion of  
\begin{equation*} 
V(\mu ^{1/4})=:(1+\Delta )^{2}\bar{U}(\mu ^{1/4})-\mu \bar{U}(\mu ^{1/4})+%
\bar{U}(\mu ^{1/4})^{3} 
\end{equation*}
in powers of $\mu ^{1/4}$ is identically 0, but we don't know whether this 
function (smooth in $\mu ^{1/4})$, which is in $\mathcal{H}_{s-4}$, is 
indeed 0. In fact we have the following 
 
\begin{theorem} 
\label{sh_exp_estim} For any even $Q\geq 8$, take $s>Q/4$. Then, $%
l=(1/2)\varphi (Q)-1$ being defined by Lemma~\ref{dioph_estimate}, the 
quasiperiodic function $\bar{U}(\mu ^{1/4l})\in {\mathcal{H}}_{s}$, with $%
s>Q/4$, defined from the series found in Theorem~\ref{gevreythm}, is 
solution of the Swift--Hohenberg PDE (\ref{eq:sh}) up to an exponentially 
small term bounded by $C(K^{\prime })e^{-\frac{1}{K^{\prime }\mu ^{1/4l}}}$ 
in $\mathcal{H}_{s-4}$, for any $K^{\prime }>K^{1/4l}$. 
\end{theorem} 
 
\begin{proof} 
The result of the Theorem follows directly from two elementary lemmas \ref%
{Product} and \ref{multiplnu^4} on Gevrey-$1$ series shown in Appendix \ref%
{app:Gevrey}, and which may be understood in the function space $\mathcal{H}%
_{s}$ instead of $%
\mathbb{C} 
$. Indeed, for $l=1$ this gives an estimate of the difference beween $V(\mu 
^{1/4})$ and the truncated Laplace transform of the left hand side of 
equation (\ref{Borel_PDEqu}) (which is then 0), taking into account of  
\begin{equation*} 
(1+\Delta )^{2}\bar{U}(\mu ^{1/4})=\frac{1}{\mu ^{1/4}}\int_{0}^{\frac{1}{%
K^{\prime }}}e^{-\frac{\zeta }{\nu }}(1+\Delta )^{2}\widehat{U}(\zeta 
)d\zeta , 
\end{equation*}
which holds in $\mathcal{H}_{s-4}$. Using Remark~\ref{alpha_diff_2}, the 
extension to larger $l$'s is left to the reader. 
\end{proof} 
 
\section{On the initial value problem} 
\label{sec:initialvalueproblem}
 
Once we know an approximate solution $\bar{U}$ of the steady PDE (\ref{eq:sh}%
), a natural question is: let us start at time $t=0$ with $U|_{t=0}=\bar{U}$, 
what can we say about the solution $U(t)$ of the initial value problem, for $%
t>0$? Let us give the following partial answer to this question: 
 
\begin{lemma} 
Assume $Q\geq 8$ and $s>Q/4$ and consider the solution $U(t)$ of the initial 
value problem (\ref{eq:shtime}) with $U|_{t=0}=\bar{U}\in \mathcal{H}_{s+4}$, 
where $\bar{U}$ is given by Theorem \ref{sh_exp_estim}. Then there are $%
\alpha $ and $C^{\prime }>0$ such that the estimate
\begin{equation*} 
||U(t)-\bar{U}||_{s}\leq C^{\prime }e^{-\frac{c}{\mu ^{1/4l}}} 
\end{equation*}
holds for $0\leq t\leq \frac{\alpha }{\mu ^{1+1/4l}}$, \ where $c$ is the 
same as in Theorem \ref{sh_exp_estim}. 
\end{lemma} 
 
\begin{proof} 
We can replace $s$ in Theorem \ref{sh_exp_estim} by $s+4$, hence we have
\begin{equation*} 
-\mathcal{L}_{0}\bar{U}+\mu \bar{U}-\bar{U}^{3}=R\in \mathcal{H}_{s}, 
\end{equation*}
with $C$ and $c>0$ such that
\begin{equation*} 
||\bar{U}||_{s+4}\leq C\sqrt{\mu },,\text{ \ }||R||_{s}\leq Ce^{-\frac{c}{%
\mu ^{1/4l}}}. 
\end{equation*}
Let us introduce the semi-group $e^{-\mathcal{L}_{0}t},t\geq 0$, defined for 
any $U\in \mathcal{H}_{s},s\geq 0$, by
\begin{equation*} 
(e^{-\mathcal{L}_{0}t}U)_{\mathbf{k}}=e^{-(1-|\mathbf{k}|^{2})^{2}t}U_{%
\mathbf{k}}. 
\end{equation*}
This semi-group is strongly continuous in $\mathcal{H}_{s}$, and bounded by 
1. Now defining $W(t)=U(t)-\bar{U}$, we have in $\mathcal{H}_{s}$  
\begin{eqnarray} 
W(t) &=&\int_{0}^{t}e^{-\mathcal{L}_{0}(t-\tau )}\{\mu W(\tau )-3\bar{U}%
^{2}W(\tau )-3\bar{U}W(\tau )^{2}-W(\tau )^{3}\}d\tau +  \notag \\ 
&&+\int_{0}^{t}e^{-\mathcal{L}_{0}(t-\tau )}Rd\tau .  \label{eq:W} 
\end{eqnarray}
We know that $W(0)=0$, and by standard arguments the solution of the initial 
value problem exists at least on a finite interval $[0,T)$ in $\mathcal{H}%
_{s}$. Let us give a more precise estimate on $W(t)$ for a part of the 
interval of time where $||W(t)||_{s}\leq C_{1}\sqrt{\mu }$ for a certain $%
C_{1}>0$. A simple estimate on (\ref{eq:W}) leads to
\begin{equation*} 
||W(t)||_{s}\leq \int_{0}^{t}\gamma _{2}||W(\tau )||_{s}d\tau +tCe^{-\frac{c%
}{\mu ^{1/4l}}}, 
\end{equation*}
with  
\begin{equation*} 
\gamma _{2}=(1+3C^{2}+3CC_{1}+C_{1}^{2})\mu . 
\end{equation*}
Then solving this inequality by Gronwall, we obtain
\begin{equation*} 
||W(t)||_{s}\leq \frac{Ce^{-\frac{c}{\mu ^{1/4l}}}}{\gamma _{2}}(e^{\gamma 
_{2}t}-1) 
\end{equation*}
which leads directly to the result of the Lemma. 
\end{proof} 
 
\appendix   
 
\section{Proof of Lemma~\protect\ref{dioph_estimate}} 
 
\label{app:dioph_estimate} 
 
We give below an elementary proof of Lemma~\ref{dioph_estimate}. 
 
The polynomial $P$ being irreducible on $\mathbb{Q}$ of degree $l+1$ and the 
polynomial $Q$ defined by  
\begin{equation*} 
Q(x)=\sum_{0\leq j\leq l}q_{j}x^{j}, 
\end{equation*}
being of degree $l$, then by the Bezout Theorem there exist two polynomials $%
A(x)$ of degree $l-1$ and $B(x)$ of degree $l$, with coefficients in $%
\mathbb{Q}$ such that  
\begin{equation} 
A(x)P(x)+B(x)Q(x)=1.  \label{Bezout} 
\end{equation}
Defining coefficients $p_{j}$, $0\leq j\leq l+1$, $a_{j}$, $0\leq j\leq l-1$ 
and $b_{j}$, $0\leq j\leq l$ of polynomials $P$, $A$ and $B$, the identity 
(\ref{Bezout}) becomes a linear system of $2l+1$ equations, of the form  
\begin{equation} 
\mathbf{M}X=\xi _{0},  \label{2l+1_system} 
\end{equation}
where the unknown is $X$ with  
\begin{equation*} 
X=\left(  
\begin{array}{c} 
a_{j-1} \\  
a_{j-2} \\  
\cdot \\  
a_{0} \\  
b_{l} \\  
b_{l-1} \\  
\cdot \\  
b_{0}%
\end{array}%
\right) ,\xi _{0}=\left(  
\begin{array}{c} 
0 \\  
0 \\  
\cdot \\  
\cdot \\  
\cdot \\  
\cdot \\  
0 \\  
1%
\end{array}%
\right) , 
\end{equation*}%
\begin{equation*} 
\mathbf{M}=\left(  
\begin{array}{ccccccccc} 
p_{l+1} & 0 & \cdot & 0 & q_{l} & 0 & \cdot & \cdot & 0 \\  
p_{l} & p_{l+1} & \cdot & \cdot & q_{l-1} & q_{l} & 0 & \cdot & \cdot \\  
\cdot & \cdot & \cdot & 0 & q_{l-2} & q_{l-1} & q_{l} & \cdot & \cdot \\  
\cdot & \cdot & \cdot & p_{l+1} & \cdot & \cdot & \cdot & \cdot & 0 \\  
p_{1} & \cdot & \cdot & p_{l} & q_{0} & \cdot & \cdot & \cdot & q_{l} \\  
p_{0} & \cdot & \cdot & p_{l-1} & 0 & q_{0} & \cdot & \cdot & q_{l-1} \\  
0 & p_{0} & \cdot & p_{l-2} & 0 & 0 & \cdot & \cdot & \cdot \\  
\cdot & \cdot & \cdot & \cdot & \cdot & \cdot & \cdot & \cdot & \cdot \\  
0 & \cdot & 0 & p_{0} & 0 & \cdot & \cdot & 0 & q_{0}%
\end{array}%
\right) . 
\end{equation*}
The $(2l+1)\times (2l+1)$ matrix $\mathbf{M}$ has integer coefficients and 
is \emph{invertible} (otherwise it would contradict the Bezout Theorem). 
Hence its determinant is integer valued and is an homogeneous polynomial of 
degree $l+1$ in $q=(q_{0},\dots ,q_{l})$. We may invert the system (\ref%
{2l+1_system}) by Cramer's formulas and we observe that the coefficients $%
b_{j}$ are rational numbers, with a common denominator of degree $l+1$ in $q$ 
and with a numerator of degree $l$ only (we replace in the determinant one 
column containing the $q_{j}$'s by $\xi _{0})$. It results that the 
polynomial $B(x)$ is the ratio of a \emph{polynomial with integer 
coefficients }$B_{0}$ \emph{of degree }$l$\emph{\ in} $q$, with an integer $%
d $, homogeneous polynomial of $q$ of degree $l+1$ and which is different 
from 0 ($\det \mathbf{M}\neq 0$). Now taking $x=\omega $ in (\ref{Bezout}) 
leads to  
\begin{equation*} 
|Q(\omega )|=\frac{d}{|B_{0}(\omega )|}, 
\end{equation*}%
and since $d\geq 1$ and the coefficients of $B_{0}$ are bounded by $%
C^{\prime }|\mathbf{q}|^{l}$, this completes the proof of Lemma~\ref%
{dioph_estimate}. 
 
\section{Proof of Lemma~\protect\ref{lem:N_k}} 
 
\label{app:lemN_k} 
 
Assertion (ii) follows from the fact that we can group the coefficients $%
m_{j}-m_{j+Q/2}=m_{j}^{\prime }$, and since in the $Q/2-$ dimensional space 
of $\{m_{j}^{\prime }$, $j=1,\dots ,Q/2\}$ the set \ $%
\sum_{j=1}^{Q/2}|m_{j}^{\prime }|=N$ is a union of $2^{Q/2}$ simplexes of 
area of order $O(N^{Q/2-1})$. To prove the part (i) (\ref{triangular}) we 
observe that  
\begin{eqnarray*} 
N_{\mathbf{k}+\mathbf{l}} &=&\min \{|m+n|;\mathbf{k}+\mathbf{l}%
=\sum_{j=1}^{Q}(m_{j}+n_{j})\mathbf{k}_{j}\} \\ 
&\leq &\min \{|m|;\mathbf{k}=\sum_{j=1}^{Q}m_{j}\mathbf{k}_{j}\}+\min \{|n|;%
\mathbf{l}=\sum_{j=1}^{Q}n_{j}\mathbf{k}_{j}\} \\ 
&\leq &N_{\mathbf{k}}+N_{\mathbf{l}}, 
\end{eqnarray*}
where  
\begin{equation*} 
N_{\mathbf{k}}=\min_{\mathbf{k}=\mathbf{k}_{\mathbf{m}}}\sum_{j=1}^{Q}m_{j}%
\mathbf{k}_{j},\text{ \ }N_{\mathbf{l}}=\min_{\mathbf{l}=\mathbf{l}_{\mathbf{%
\ n}}}\sum_{j=1}^{Q}n_{j}\mathbf{k}_{j}. 
\end{equation*}
We notice that $N_{\mathbf{0}}=0$, and $N_{-\mathbf{k}}=N_{\mathbf{k}}$ 
(each $m_{j}^{\prime }$ for $\mathbf{k}$ is just the opposite for $-\mathbf{k%
}$); we deduce that inequality (\ref{triangular}) may be strict, since  
\begin{equation*} 
0=N_{\mathbf{0}}=N_{\mathbf{k}-\mathbf{k}}<N_{-\mathbf{k}}+N_{\mathbf{k}%
}=2N_{\mathbf{k}}. 
\end{equation*}
The last inequality (\ref{compare_k_and_N_k}) is easily deduced from  
\begin{equation*} 
\mathbf{k}=\sum_{j=1}^{Q}m_{j}\mathbf{k}_{j} 
\end{equation*}
where $\{m_{j}\}$ gives precisely the \textquotedblleft 
norm\textquotedblright\ $N_{\mathbf{k}};$ which implies (since $|\mathbf{k}%
_{j}|=1)$  
\begin{equation*} 
|\mathbf{k}|\leq \sum_{j=1}^{Q}|m_{j}|=N_{\mathbf{k}}, 
\end{equation*}
and the Lemma is proved. 
 
\section{Proof of Lemma~\protect\ref{algebra} \label{app:algebra}} 
 
Let $u\in {\mathcal{H}}_{s}$, then by Cauchy--Schwarz inequality in $%
l^{2}(\Gamma )$ ($\Gamma $ is countable) we have  
\begin{eqnarray*} 
\left\vert \sum_{\mathbf{k}\in \Gamma }u_{\mathbf{k}}e^{i\mathbf{k}\cdot  
\mathbf{x}}\right\vert ^{2} &\leq &\left( \sum_{\mathbf{k}\in \Gamma }(1+N_{  
\mathbf{k}}{}^{2})^{s}|u_{\mathbf{k}}|^{2}\right) \sum_{\mathbf{k}\in \Gamma 
}\frac{1}{(1+N_{\mathbf{k}}{}^{2})^{s}} \\ 
&\leq &||u||_{{\mathcal{H}}_{s}}^{2}\sum_{\mathbf{k}\in \Gamma }\frac{1}{ 
(1+N_{\mathbf{k}}{}^{2})^{s}}. 
\end{eqnarray*} 
Now by (\ref{card_k}) we have the following estimate  
\begin{equation*} 
\sum_{\mathbf{k}\in \Gamma }\frac{1}{(1+N_{\mathbf{k}}{}^{2})^{s}}\leq 
c_{1}(Q)\sum_{n\mathbf{\in\mathbb{N}}}\frac{n^{Q/2-1}}{(1+n^{2})^{s}} 
\end{equation*} 
which is bounded when $s>Q/4$. Hence for $s>Q/4$ the series $\sum_{\mathbf{k} 
\in \Gamma }u_{\mathbf{k}}e^{i\mathbf{k}\cdot \mathbf{x}}$ converges 
absolutely and represents a continuous quasiperiodic function, the norm 
(uniform norm) of which being bounded as soon as the norm in ${\mathcal{H}} 
_{s}$ is bounded. We may proceed in the same way for the derivatives in 
using (\ref{compare_k_and_N_k}), and show that the series  
\begin{equation*} 
\sum_{\mathbf{k}\in \Gamma }|\mathbf{k}|^{l}u_{\mathbf{k}}e^{i\mathbf{k} 
\cdot \mathbf{x}} 
\end{equation*} 
is absolutely convergent for $s>Q/4+l$. This ends the proof of the last 
assertion of the Lemma. Let us now prove the first assertion which is 
necessary for our nonlinear problem. 
 
\emph{First step:} We first use the following inequality due to (\ref%
{triangular})  
\begin{equation*} 
(1+N_{\mathbf{k}+\mathbf{k}^{\prime }}^{2})^{s/2}\leq 2^{s-1}\left\{ (1+N_{  
\mathbf{k}}^{2})^{s/2}+(1+N_{\mathbf{k}^{\prime }}^{2})^{s/2}\right\} 
\end{equation*} 
valid for any $s\geq 1$, because of (\ref{triangular}) and a simple 
convexity argument (this inequality is in fact valid for $s>0)$. Then the 
following decomposition holds  
\begin{equation*} 
\sum_{\mathbf{K}}\left\vert \sum_{\mathbf{k}+\mathbf{k}^{\prime }=\mathbf{K} 
}u_{\mathbf{k}}v_{\mathbf{k}^{\prime }}\right\vert ^{2}(1+N_{\mathbf{K} 
}^{2})^{s}\leq 2^{2s-1}(S_{1}+S_{2}) 
\end{equation*} 
with  
\begin{eqnarray*} 
S_{1} &=&\sum_{\mathbf{K}}\left\vert \sum_{\mathbf{k}+\mathbf{k}^{\prime }=  
\mathbf{K}}u_{\mathbf{k}}v_{\mathbf{k}^{\prime }}\right\vert ^{2}(1+N_{  
\mathbf{k}}^{2})^{s} \\ 
S_{2} &=&\sum_{\mathbf{K}}\left\vert \sum_{\mathbf{k}+\mathbf{k}^{\prime }=  
\mathbf{K}}u_{\mathbf{k}}v_{\mathbf{k}^{\prime }}\right\vert ^{2}(1+N_{  
\mathbf{k}^{\prime }}^{2})^{s}. 
\end{eqnarray*} 
For symmetry reasons in the space $(\mathbf{k},\mathbf{k}^{\prime })$, it is 
then sufficient to estimate $S_{1}$. Let us split the bracket in the sum $%
S_{1}$ into two terms: a sum $S_{1}^{\prime }$ containing $(\mathbf{k},  
\mathbf{k}^{\prime })$ such that  
\begin{equation*} 
N_{\mathbf{k}}\leq 3N_{\mathbf{k}^{\prime }}, 
\end{equation*} 
and a sum $S_{1}^{\prime \prime }$ containing $(\mathbf{k},\mathbf{k} 
^{\prime })$ such that $N_{\mathbf{k}}>3N_{\mathbf{k}^{\prime }}$. Hence we 
have now  
\begin{equation*} 
S_{1}\leq 2(S_{1}^{\prime }+S_{1}^{\prime \prime }) 
\end{equation*} 
with  
\begin{eqnarray*} 
S_{1}^{\prime } &=&\sum_{\mathbf{K}}\left\vert \sum_{\substack{ \mathbf{k}+%
\mathbf{k}^{\prime }=\mathbf{K},  \\ N_{\mathbf{k}}\leq 3N_{\mathbf{k}%
^{\prime }}}} u_{\mathbf{k}}v_{\mathbf{k}^{\prime }}\right\vert ^{2}(1+N_{%
\mathbf{k} }^{2})^{s}, \\ 
S_{1}^{\prime \prime } &=&\sum_{\mathbf{K}}\left\vert \sum_{\substack{  
\mathbf{k}+\mathbf{k}^{\prime }=\mathbf{K},  \\ N_{\mathbf{k}}>3N_{\mathbf{k}%
^{\prime }}}} u_{\mathbf{k}}v_{\mathbf{k}^{\prime }}\right\vert^{2} (1+N_{%
\mathbf{k}}^{2})^{s}. 
\end{eqnarray*} 
To estimate $S_{1}^{\prime }$ we use (\ref{triangular}) which gives $N_{  
\mathbf{K}}\leq 4N_{\mathbf{k}^{\prime }}$, hence  
\begin{equation*} 
\frac{1}{1+N_{\mathbf{k}^{\prime }}^{2}}\leq \frac{16}{1+N_{\mathbf{K}}^{2}}, 
\end{equation*} 
and, in using again Cauchy--Schwarz  
\begin{eqnarray*} 
\sum_{\substack{ \mathbf{k}+\mathbf{k}^{\prime }=\mathbf{K},  \\ N_{\mathbf{k%
}}\leq 3N_{\mathbf{k}^{\prime }}}} |u_{\mathbf{k}}v_{\mathbf{k}^{\prime 
}}|(1+N_{\mathbf{k}}^{2})^{s/2} &\leq & \sum_{\substack{ \mathbf{k}+\mathbf{k%
}^{\prime }=\mathbf{K},  \\ N_{\mathbf{k}}\leq 3N_{\mathbf{k}^{\prime }}}} 
4^{s}|u_{\mathbf{k}}v_{\mathbf{k}^{\prime }}| \frac{(1+N_{\mathbf{k}%
}^{2})^{s/2}(1+N_{\mathbf{k}^{\prime }}^{2})^{s/2}} {(1+N_{\mathbf{K}%
}^{2})^{s/2}} \\ 
&\leq &\frac{4^{s}}{(1+N_{\mathbf{K}}^{2})^{s/2}} ||u||_{{\mathcal{H}}%
_{s}}||v||_{{\mathcal{H}}_{s}}. 
\end{eqnarray*} 
It results that  
\begin{equation*} 
S_{1}^{\prime }\leq ||u||_{{\mathcal{H}}_{s}}^{2}||v||_{{\mathcal{H}} 
_{s}}^{2}\sum_{\mathbf{K}}\frac{4^{2s}}{(1+N_{\mathbf{K}}^{2})^{s}} 
\end{equation*} 
which, for $s>Q/4$ leads to  
\begin{equation*} 
S_{1}^{\prime }\leq C||u||_{{\mathcal{H}}_{s}}^{2}||v||_{{\mathcal{H}} 
_{s}}^{2}. 
\end{equation*} 
 
\emph{Second step}: We now find a bound for $S_{1}^{\prime \prime }$, which 
is more technical, since we split this sum into packets of increasing 
lengths. 
 
Let us define  
\begin{equation*} 
\Delta _{p}u=\sum_{2^{p}\leq N_{\mathbf{k}}<2^{p+1}}u_{\mathbf{k}}e^{i  
\mathbf{k}\cdot \mathbf{x}},\text{ \ }\Delta _{-1}u=u_{\mathbf{0}}. 
\end{equation*} 
It is clear that for $s>Q/4$ (the series is absolutely convergent)  
\begin{equation*} 
u=\sum_{p=-1}^{\infty }\Delta _{p}u. 
\end{equation*} 
Moreover, it is clear from the definition that the norm of $u\in {\mathcal{H} 
}_{s}$ is equivalent to  
\begin{equation*} 
\left( \sum_{p=-1}^{\infty }2^{2ps}||\Delta _{p}u||_{0}^{2}\right) ^{1/2}. 
\end{equation*} 
To estimate the sum $S_{1}^{\prime \prime }$, we notice that in the product $%
uv$ the terms $\Delta _{p}u\Delta _{q}v$ only take into account the 
wavevectors $\mathbf{k}$ and $\mathbf{k}^{\prime }$ such that  
\begin{equation*} 
2^{p}\leq N_{\mathbf{k}}<2^{p+1},\text{ \ }2^{q}\leq N_{\mathbf{k}^{\prime 
}}<2^{q+1},\text{ \ }N_{\mathbf{k}}>3N_{\mathbf{k}^{\prime }}. 
\end{equation*} 
This implies  
\begin{equation*} 
N_{\mathbf{k}^{\prime }}<2^{p},\qquad 2^{q+1}<N_{\mathbf{k}}, 
\end{equation*} 
hence in $S_{1}^{\prime \prime }$  
\begin{equation*} 
\Delta _{p}u\Delta _{q}v=0,\text{ for }p\leq q. 
\end{equation*} 
Now, we use (for the sum in $S_{1}^{\prime \prime })$  
\begin{equation*} 
\frac{2}{3}N_{\mathbf{k}}\leq N_{\mathbf{K}} 
\end{equation*} 
\begin{equation*} 
S_{1}^{\prime \prime }\leq (\frac{2}{3})^{2s}\sum_{\mathbf{K}}\left\vert 
\sum _{\substack{ \mathbf{k}+\mathbf{k}^{\prime }=\mathbf{K},  \\ N_{\mathbf{%
k}}>3N_{\mathbf{k}^{\prime }}}} u_{\mathbf{k}}v_{\mathbf{k}^{\prime 
}}\right\vert ^{2}(1+N_{\mathbf{K}}^{2})^{s} 
\end{equation*} 
and the right hand side is the square of the norm of the product $uv$ 
computed on terms such that $N_{\mathbf{k}}>3N_{\mathbf{k}^{\prime }}$, $%
\mathbf{k}+\mathbf{k}^{\prime }=\mathbf{K}$. We now use the equivalent norm 
defined above with the decomposition in packets, hence  
\begin{equation*} 
S_{1}^{\prime \prime }\leq C\sum_{j=-1}^{\infty }2^{2js}\left\Vert \Delta 
_{j}\left( \sum_{p\geq 0}\left( \sum_{q=-1}^{p-1}\Delta _{q}v\right) \Delta 
_{p}u\right) \right\Vert _{0}^{2}. 
\end{equation*} 
Let us define $S_{p-1}v=\sum_{q=-1}^{p-1}\Delta _{q}v$, then we have  
\begin{equation*} 
\Delta _{j}\left( \sum_{p}S_{p-1}v\Delta _{p}u\right) 
=\sum_{p=j-1}^{j+1}\Delta _{j}(S_{p-1}v\Delta _{p}u)2^{ps}2^{-ps} 
\end{equation*} 
hence by Cauchy--Schwarz  
\begin{equation*} 
2^{2js}\left\Vert \Delta _{j}\left( \sum_{p}S_{p-1}v\Delta _{p}u\right) 
\right\Vert _{0}^{2}\leq \left( \sum_{p=j-1}^{j+1}2^{2(j-p)s}\right) 
\sum_{p=j-1}^{j+1}2^{2ps}\left\Vert \Delta _{j}(S_{p-1}v\Delta 
_{p}u)\right\Vert _{0}^{2} 
\end{equation*} 
Now  
\begin{equation*} 
\left\Vert S_{p-1}v\Delta _{p}u\right\Vert _{0}^{2}=\sum_{\mathbf{K}}|\sum_{  
\mathbf{k}+\mathbf{k}^{\prime }=\mathbf{K},0\leq N_{\mathbf{k}^{\prime 
}}<2^{p}\leq N_{\mathbf{k}}<2^{p+1}}u_{\mathbf{k}}v_{\mathbf{k}^{\prime 
}}|^{2} 
\end{equation*} 
and a classical computation (convolution $l^{1}\ast l^{2})$ using 
Cauchy--Schwarz gives  
\begin{eqnarray*} 
\sum_{\mathbf{K}}|\sum_{\mathbf{k}+\mathbf{k}^{\prime }=\mathbf{K}}u_{  
\mathbf{k}}v_{\mathbf{k}^{\prime }}|^{2} &\leq &\sum_{\mathbf{K}}\left\{ 
(\sum_{\mathbf{k}+\mathbf{k}^{\prime }=\mathbf{K}}|v_{\mathbf{k}^{\prime 
}}||u_{\mathbf{k}}|^{2})(\sum_{\mathbf{k}^{\prime }}|v_{\mathbf{k}^{\prime 
}}|)\right\} \\ 
&\leq &(\sum_{\mathbf{k}^{\prime }}|v_{\mathbf{k}^{\prime }}|)(\sum_{\mathbf{%
\ k}}\sum_{\mathbf{K}}|v_{\mathbf{K}-\mathbf{k}}||u_{\mathbf{k}}|^{2}) \\ 
&\leq &\left( (\sum_{\mathbf{k}^{\prime }}|v_{\mathbf{k}^{\prime }}|)\right) 
^{2}\sum_{\mathbf{k}}|u_{\mathbf{k}}|^{2}) 
\end{eqnarray*} 
which leads to  
\begin{equation*} 
\left\Vert S_{p-1}v\Delta _{p}u\right\Vert _{0}^{2}\leq ||\Delta 
_{p}u||_{0}^{2}\left( (\sum_{\mathbf{k}^{\prime }}|v_{\mathbf{k}^{\prime 
}}|)\right) ^{2} 
\end{equation*} 
and since the series$\sum |v_{\mathbf{k}^{\prime }}|\leq c||v||_{\mathcal{H} 
_{s}}$ for $s>Q/4$, as shown at the beginning of the proof of Lemma~\ref%
{algebra}, we have  
\begin{equation*} 
\left\Vert S_{p-1}v\Delta _{p}u\right\Vert _{0}^{2}\leq C||\Delta 
_{p}u||_{0}^{2}||v||_{{\mathcal{H}}_{s}}^{2}. 
\end{equation*} 
Finally, we obtain  
\begin{eqnarray*} 
\sum_{p=j-1}^{j+1}2^{2ps}\left\Vert \Delta _{j}(S_{p-1}v\Delta 
_{p}u)\right\Vert _{0}^{2} &\leq &\sum_{p=j-1}^{j+1}2^{2ps}\left\Vert 
S_{p-1}v\Delta _{p}u\right\Vert _{0}^{2} \\ 
&\leq &C^{\prime }||v||_{{\mathcal{H}}_{s}}^{2}\sum_{p=j-1}^{j+1}2^{2ps}|| 
\Delta _{p}u||_{0}^{2}, 
\end{eqnarray*} 
and  
\begin{equation*} 
2^{2js}\left\Vert \Delta _{j}\left( \sum_{p}S_{p-1}v\Delta _{p}u\right) 
\right\Vert _{0}^{2}\leq C^{\prime \prime }||v||_{\mathcal{H} 
_{s}}^{2}\sum_{p=j-1}^{j+1}2^{2ps}||\Delta _{p}u||_{0}^{2}, 
\end{equation*} 
hence  
\begin{eqnarray*} 
S_{1}^{\prime \prime } &\leq &3C^{^{\prime \prime }}||v||_{\mathcal{H} 
_{s}}^{2}\sum_{p=-1}^{\infty }2^{2ps}||\Delta _{p}u||_{0}^{2} \\ 
&\leq &C_{1}||u||_{{\mathcal{H}}_{s}}^{2}||v||_{{\mathcal{H}}_{s}}^{2} 
\end{eqnarray*} 
and Lemma~\ref{algebra} is proved. 
 
\section{Proof of Lemma~\protect\ref{factorials} \label{app:factorials}} 
 
Let us define the two sums  
\begin{eqnarray*} 
\Pi _{2,n} &=&\sum_{k=0}^{n}(k!(n-k)!)^{4l} \\ 
\Pi _{2,n}^{\prime } &=&\sum_{k=1}^{n-1}(k!(n-k)!)^{4l} 
\end{eqnarray*}
we have already  
\begin{eqnarray*} 
\Pi _{2,0} &=&1,\text{ \ }\Pi _{2,1}=2,\qquad \Pi _{2,2}=(2+\frac{1}{2^{4l}}%
)(2!)^{4l}, \\ 
\Pi _{2,2}^{\prime } &=&1,\text{ \ }\Pi _{2,3}^{\prime }=2(2!)^{4l}, 
\end{eqnarray*}
which shows that $\Pi _{2,n}\leq (2+\frac{1}{16})(n!)^{4}$ for $n=0,1,2$, 
and $l\geq 1$. Now we have for $n\geq 2$  
\begin{eqnarray*} 
\frac{\Pi _{2,n+1}}{((n+1)!)^{4l}}-\frac{\Pi _{2,n}}{(n!)^{4l}} 
&=&\sum_{k=2}^{n-2}\left( \frac{k!(n-k)!}{n!}\right) ^{4l}\left\{ \left(  
\frac{n+1-k}{n+1}\right) ^{4l}-1\right\} + \\ 
&&+\frac{2}{(n+1)^{4l}}-\frac{2}{n^{4l}}+\frac{2^{4l}}{(n(n+1))^{4l}}, 
\end{eqnarray*}
and since $n^{4l}-(n+1)^{4l}+2^{4l-1}<0$ for $n\geq 1$ the above right hand 
side terms are negative. It results that for $n\geq 2$  
\begin{equation*} 
\Pi _{2,n+1}\leq \left( \frac{(n+1)!}{n!}\right) ^{4l}\Pi _{n,2}, 
\end{equation*}
hence  
\begin{equation} 
\Pi _{2,n}\leq (2+\frac{1}{16})(n!)^{4l},\qquad n\geq 0. 
\label{estim_pi_2,n} 
\end{equation}
In the same way  
\begin{equation*} 
\frac{\Pi _{2,n+1}^{\prime }}{(n!)^{4l}}-\frac{\Pi _{2,n}^{\prime }}{%
((n-1)!)^{4l}}=\sum_{k=2}^{n-2}\left( \frac{k!(n-k)!}{n!}\right) 
^{4l}\left\{ \left( \frac{n+1-k}{n+1}\right) ^{4l}-1\right\} +\frac{2^{4l}}{%
n^{4l}}, 
\end{equation*}
hence for $n\geq 2$  
\begin{equation*} 
\frac{\Pi _{2,n+1}^{\prime }}{(n!)^{4l}}\leq \frac{\Pi _{2,n}^{\prime }}{%
((n-1)!)^{4l}}+\frac{2^{4l}}{n^{4l}}, 
\end{equation*}
and  
\begin{eqnarray*} 
\frac{\Pi _{2,n}^{\prime }}{((n-1!)^{4l}} &\leq &2^{4l}\left( \frac{1}{%
(n-1)^{4l}}+\dots +\frac{1}{2^{4l}}\right) +\Pi _{2,2}^{\prime } \\ 
&\leq &2+2^{4l}\left( \frac{1}{(n-1)^{4l}}+\dots +\frac{1}{3^{4l}}\right) \\ 
&\leq &2+\frac{2}{4l-1}\leq 3. 
\end{eqnarray*}
Finally  
\begin{equation} 
\Pi _{2,n}^{\prime }\leq 3((n-1!)^{4l}\text{ for }n\geq 2. 
\label{estim_pi'_2,n} 
\end{equation}
Consider now $\Pi _{3,n}$ defined by  
\begin{equation*} 
\Pi _{3,n}=\sum_{\substack{ k+l+r=n  \\ k,l,r\geq 0}}(k!l!r!)^{4l}. 
\end{equation*}
We already have  
\begin{equation*} 
\Pi _{3,0}=1,\text{ \ }\Pi _{3,1}=3,\text{ \ \ }\Pi _{3,2}=(3+\frac{3}{2^{4l}%
})(2!)^{4l}\leq 4(2!)^{4l}, 
\end{equation*}
In splitting the sum we obtain easily for $n\geq 3$  
\begin{eqnarray*} 
\Pi _{3,n} &=&\Pi _{2,n}+(n!)^{4l}+\sum_{r=1}^{n-1}(r!)^{4l}\Pi _{2,n-r} \\ 
&\leq &(3+\frac{1}{16})(n!)^{4l}+(2+\frac{1}{16})\Pi _{2,n}^{\prime } \\ 
&\leq &(n!)^{4l}(3+\frac{1}{16}+3(2+\frac{1}{16})\frac{1}{n^{4l}}) \\ 
&\leq &(3+\frac{3}{16}+\frac{9}{3^{4}})(n!)^{4l}\leq 4(n!)^{4l}. 
\end{eqnarray*}
Hence  
\begin{equation} 
\Pi _{3,n}\leq 4(n!)^{4l}  \label{estim_pi_3,n} 
\end{equation}
holds for any $n\geq 0$. Consider now $\Pi _{3,n}^{\prime }$ defined for $%
n\geq 2$ by  
\begin{equation*} 
\Pi _{3,n}^{\prime }=\sum_{\substack{ k+l+r=n  \\ 0\leq k,l,r\leq n-1}}%
(k!l!r!)^{4l}. 
\end{equation*}
We already have  
\begin{equation*} 
\Pi _{3,2}^{\prime }=1, 
\end{equation*}
and for $n\geq 3$, we obtain in the same way  
\begin{eqnarray} 
\Pi _{3,n}^{\prime } &=&\Pi _{2,n}^{\prime }+\sum_{r=1}^{n-1}(r!)^{4l}\Pi 
_{2,n-r}  \notag \\ 
&\leq &\left( 3+3(2+\frac{1}{16})\right) (n-1!)^{4l}  \notag \\ 
&\leq &10(n-1!)^{4l}.  \label{estim_pi'_3,n} 
\end{eqnarray}
Hence, with estimates (\ref{estim_pi_3,n}) and (\ref{estim_pi'_3,n}), 
Lemma~\ref{factorials} is proved. 
 
\section{Lemmas on Gevrey-$1$ series} 
 
\label{app:Gevrey} 
 
Below we give elementary proofs of two useful lemmas. The interested reader 
will find more general results in \cite{Ramis1996} and \cite{Marco2003}. 
 
In the following we denote by $\mathcal{L}_{K^{\prime }}$ the linear 
operator defined for analytic functions $v$ on the disc $\{|z|<1/K_{1}\}$ by
\begin{equation*} 
(\mathcal{L}_{K^{\prime }}v)(\nu )=\frac{1}{\nu }\int_{0}^{1/K^{\prime }}e^{-%
\frac{z}{\nu }}v(z)dz,\text{ \ }K^{\prime }>K_{1}. 
\end{equation*}
We also use the notations  
\begin{equation*} 
||v||_{0,K^{\prime }}=\sup_{z\in (0,1/K^{\prime })}|v(z)|,\text{ }%
||v||_{1,K^{\prime }}=\sup_{z\in (0,1/K^{\prime })}|v^{\prime }(z)|, 
\end{equation*}
and when $v(0)=0$, we notice that (integrating by parts for the second 
estimate)
\begin{eqnarray} 
\left\vert (\mathcal{L}_{K^{\prime }}v)(\nu )\right\vert &\leq 
&||v||_{0,K^{\prime }},  \label{simplEstim} \\ 
\left\vert (\mathcal{L}_{K^{\prime }}v)(\nu )-\int_{0}^{1/K^{\prime }}e^{-%
\frac{z}{\nu }}v^{\prime }(z)dz\right\vert &\leq &e^{-\frac{1}{K^{\prime 
}\nu }}||v||_{0,K^{\prime }}.  \notag 
\end{eqnarray}
Then we have the following Lemmas giving estimates of the commutator of $%
\mathcal{L}_{K^{\prime }}\circ \mathcal{B}$ (where $\mathcal{B}$ is the 
Borel transform) with the multiplication by $\nu ^{4}$ and with the mapping $%
u\mapsto u^{3}$ in the space of Gevrey series. 
 
\begin{lemma} 
\label{Product} Assume that $u(\nu )$ is a Gevrey-$1$ series, with $u_{0}=0$, 
then for $\nu <1/K^{\prime }$
\begin{equation*} 
\left\vert \left( \mathcal{L}_{K^{\prime }}\widehat{u^{3}}\right) (\nu 
)-\left( \mathcal{L}_{K^{\prime }}\widehat{u}\right) ^{3}(\nu )\right\vert 
\leq \frac{e^{-\frac{1}{K^{\prime }\nu }}}{(K^{\prime }\nu )^{3}}||\widehat{u%
}||_{0,K^{\prime }}(||\widehat{u}||_{0,K^{\prime }}+\nu ||\widehat{u}%
||_{1,K^{\prime }})^{2}. 
\end{equation*}
For any given Gevrey-$1$ series $u$, with $u_{0}=0$, there is $C(K^{\prime 
})>0 $ such that for $\nu <\nu _{0}(K^{\prime })$ we have the estimate
\begin{equation*} 
\left\vert \left( \mathcal{L}_{K^{\prime }}\widehat{u^{3}}\right) (\nu 
)-\left( \mathcal{L}_{K^{\prime }}\widehat{u}\right) ^{3}(\nu )\right\vert 
\leq C(K^{\prime })e^{-\frac{1}{K^{\prime }\nu }},\text{ \ }K^{\prime 
}>K_{1}. 
\end{equation*} 
\end{lemma} 
 
\begin{lemma} 
\label{multiplnu^4} Assume that $u(\nu )$ is a Gevrey-$1$ series, with $%
u_{0}=0$, then for $\nu <1/K^{\prime }$ there exists $C(K^{\prime })$ such 
that
\begin{equation*} 
\left\vert (\mathcal{L}_{K^{\prime }}\mathcal{K}\widehat{u})(\nu )-\nu ^{4}%
\mathcal{L}_{K^{\prime }}\widehat{u}\right\vert \leq C(K^{\prime })||%
\widehat{u}||_{0,K^{\prime }}e^{-\frac{1}{K^{\prime }\nu }}. 
\end{equation*} 
\end{lemma} 
 
\begin{proof}[Proof of Lemma \protect\ref{Product}] 
From the identity
\begin{equation*} 
\int_{0}^{z}\left( \int_{0}^{z_{1}}\frac{%
z_{1}^{k-1}z_{2}^{m-1}(z-z_{1}-z_{2})^{l}}{(k-1)!(m-1)!l!}dz_{2}\right) 
dz_{1}=\frac{z^{k+m+l}}{(k+m+l)!}, 
\end{equation*}
from the definition (\ref{defconvol}) of the convolution product, and from 
the analyticity of $\widehat{u}$ in the disc $\{|z|<1/K_{1}\}$, we have  
\begin{eqnarray*} 
\left( \mathcal{L}_{K^{\prime }}(\widehat{u}\ast \widehat{u}\ast \widehat{u}%
)\right) (\nu ) &=&\left( \mathcal{L}_{K^{\prime }}\widehat{u^{3}}\right) 
(\nu )= \\ 
&=&\frac{1}{\nu }\int_{0}^{1/K^{\prime }}e^{-\frac{z}{\nu }}\left( 
\int_{0}^{z}\left( \int_{0}^{z_{1}}\widehat{u}^{\prime }(z_{1})\widehat{u}%
^{\prime }(z_{2})\widehat{u}(z-z_{1}-z_{2})dz_{2}\right) dz_{1}\right) dz. 
\end{eqnarray*}
By Fubini's theorem and a simple change of variables, we obtain
\begin{equation} 
\left( \mathcal{L}_{K^{\prime }}\widehat{u^{3}}\right) (\nu )=\frac{1}{\nu }%
\int_{D_{K^{\prime }}}e^{-\frac{z_{1}+z_{2}+z_{3}}{\nu }}\widehat{u}^{\prime 
}(z_{1})\widehat{u}^{\prime }(z_{2})\widehat{u}(z_{3})dz_{1}dz_{2}dz_{3} 
\label{Laplacubic} 
\end{equation}
where $D_{K^{\prime }}=\{z_{1},z_{2},z_{3}>0;z_{1}+z_{2}+z_{3}<1/K^{\prime 
}\}$. Now, we have
\begin{equation*} 
\left( \mathcal{L}_{K^{\prime }}\widehat{u}\right) ^{3}(\nu )=\frac{1}{\nu 
^{3}}\int_{(0,1/K^{\prime })^{3}}e^{-\frac{z_{1}+z_{2}+z_{3}}{\nu }}\widehat{%
u}(z_{1})\widehat{u}(z_{2})\widehat{u}(z_{3})dz_{1}dz_{2}dz_{3}, 
\end{equation*}
and from (\ref{simplEstim}) we obtain
\begin{equation*} 
\left\vert \left( \mathcal{L}_{K^{\prime }}\widehat{u}\right) ^{3}(\nu )-%
\frac{1}{\nu }\int_{(0,1/K^{\prime })^{3}}e^{-\frac{z_{1}+z_{2}+z_{3}}{\nu }}%
\widehat{u}^{\prime }(z_{1})\widehat{u}^{\prime }(z_{2})\widehat{u}%
(z_{3})dz_{1}dz_{2}dz_{3}\right\vert \leq 
\end{equation*}
\begin{equation} 
\leq e^{-\frac{1}{K^{\prime }\nu }}||\widehat{u}||_{0,K^{\prime }}^{2}(||%
\widehat{u}||_{0,K^{\prime }}+\nu ||\widehat{u}||_{1,K^{\prime }}). 
\label{estim2} 
\end{equation}
Now, we observe that $(0,1/K^{\prime })^{3}\backslash D_{K^{\prime }}$ is 
such that $z_{1}+z_{2}+z_{3}>1/K^{\prime }$, hence
\begin{equation} 
\left\vert \frac{1}{\nu }\int_{(0,1/K^{\prime })^{3}\backslash D_{K^{\prime 
}}}e^{-\frac{z_{1}+z_{2}+z_{3}}{\nu }}\widehat{u}^{\prime }(z_{1})\widehat{u}%
^{\prime }(z_{2})\widehat{u}(z_{3})dz_{1}dz_{2}dz_{3}\right\vert \leq \frac{%
e^{-\frac{1}{K^{\prime }\nu }}}{\nu ^{3}K^{\prime 3}}||\widehat{u}%
||_{0,K^{\prime }}(\nu ||\widehat{u}||_{1,K^{\prime }})^{2}. 
\label{remainderestimate} 
\end{equation}
Collecting (\ref{Laplacubic}), (\ref{estim2}) and (\ref{remainderestimate}) 
the first result of Lemma \ref{Product} is proved. Notice that by choosing $%
K^{\prime \prime }>K^{\prime }$, then for $\nu $ small enough $\frac{e^{-%
\frac{1}{K^{\prime }\nu }}}{\nu ^{3}K^{\prime 3}}\leq e^{-\frac{1}{K^{\prime 
\prime }\nu }}$. Since $K^{\prime }$ is chosen arbitrarily larger than $%
K_{1}$, we can assert that $u$ being given, there is $C(K^{\prime })$ such 
that
\begin{equation*} 
\left\vert \left( \mathcal{L}_{K^{\prime }}\widehat{u^{3}}\right) (\nu 
)-\left( \mathcal{L}_{K^{\prime }}\widehat{u}\right) ^{3}(\nu )\right\vert 
\leq C(K^{\prime })e^{-\frac{1}{K^{\prime }\nu }},\text{ \ }K^{\prime 
}>K_{1}. 
\end{equation*} 
\end{proof} 
 
\begin{proof}[Proof of Lemma \protect\ref{multiplnu^4}] 
By integrating by parts, we obtain
\begin{eqnarray*} 
(\mathcal{L}_{K^{\prime }}\mathcal{K}\widehat{u})(\nu ) &=&-e^{-\frac{1}{%
K^{\prime }\nu }}\left[ (\mathcal{K}\widehat{u})+\nu (\mathcal{K}\widehat{u}%
)^{\prime }+\nu ^{2}(\mathcal{K}\widehat{u})^{\prime \prime }+\nu ^{3}(%
\mathcal{K}\widehat{u})^{\prime \prime \prime }\right] |_{1/K^{\prime }}+ \\ 
&&+\nu ^{4}(\mathcal{L}_{K^{\prime }}\widehat{u})(\nu ). 
\end{eqnarray*}
Hence
\begin{equation*} 
\left\vert (\mathcal{L}_{K^{\prime }}\mathcal{K}\widehat{u})(\nu )-\nu ^{4}(%
\mathcal{L}_{K^{\prime }}\widehat{u})(\nu )\right\vert \leq e^{-\frac{1}{%
K^{\prime }\nu }}||\widehat{u}||_{0,K^{\prime }}\left\{ \frac{\nu ^{3}}{%
K^{\prime }}+\frac{\nu ^{2}}{2K^{\prime 2}}+\frac{\nu }{6K^{\prime 3}}+\frac{%
1}{24K^{\prime 4}}\right\} 
\end{equation*}
which proves Lemma \ref{multiplnu^4}. 
\end{proof} 
 
 
\bibliographystyle{plain} 
\bibliography{IR4} 
 
 
%
%
%
%
 
 
\end{document}